\documentclass{mscs}

\usepackage{xcolor}
\usepackage{amssymb}
\usepackage{MnSymbol}
\usepackage{verbatim}

\newtheorem{theorem}{Theorem}[section]
\newtheorem{lemma}[theorem]{Lemma}
\newtheorem{definition}[theorem]{Definition}
\newtheorem{example}[theorem]{Example}
\newtheorem{corollary}[theorem]{Corollary}
\newtheorem{proposition}[theorem]{Proposition}
\newtheorem{remark}[theorem]{Remark}

\newtheorem{fact}[theorem]{Fact}
\newtheorem{note}[theorem]{Note}

\renewcommand{\iff}{\Leftrightarrow}
\newcommand{\tbSigma}{\mathbf{\Sigma}}
\newcommand{\tbPi}{\mathbf{\Pi}}
\newcommand{\tbDelta}{\mathbf{\Delta}}
\newcommand{\tbGamma}{\mathbf{\Gamma}}
\newcommand{\tbD}{\mathbf{D}}
\newcommand{\tlD}{\mbox{D}}
\newcommand{\bGd}{\mathbf{G}_\delta}
\newcommand{\bFs}{\mathbf{F}_\sigma}
\newcommand{\N}{\mathbb{N}}
\newcommand{\Z}{\mathbb{Z}}
\newcommand{\Q}{\mathbb{Q}}
\newcommand{\R}{\mathbb{R}}

\newcommand{\cantor}{\alphabet^\omega}
\newcommand{\alphabet}{{\bf 2}}
\newcommand{\mix}{\alphabet^{\leq \omega}}
\newcommand{\Bmix}{\N^{\leq \omega}}
\newcommand{\Baire}{\N^{\omega}}
\newcommand{\nil}{\textit{nil}}
\newcommand{\uup}{\twoheaduparrow\!\!}
\newcommand{\ddown}{\twoheaddownarrow\!\!}
\newcommand{\co}{\textit{co-}}
\newcommand{\CK}{\omega_1^{\textit{\scriptsize CK}}}

\renewcommand{\P}{\+P}
\newcommand{\PN}{\+P\N}
\newcommand{\FIN}{\+P_{\!<\omega}\N}

\begin{document}

\title[Wadge Hardness in Scott Spaces and Its Effectivization]
 {Wadge Hardness in Scott Spaces\\ and Its Effectivization}

\author[Ver\'onica Becher \& Serge Grigorieff]
{Ver\'onica Becher $^{1}$\thanks{Members of the Laboratoire Internationale Associ\'e INFINIS, Universidad de Buenos Aires-Universit\'e Paris Diderot-Paris 7. This research was partially done whilst the first author was a visiting fellow at the Isaac Newton Institute for Mathematical Sciences in the programme `Semantics \& Syntax'.
}
and Serge Grigorieff $^{2,\dagger}$\\
$^1$ FCEyN,Universidad de Buenos Aires \& CONICET, Argentina.
\addressbreak
$^2$ LIAFA, CNRS \& Universit\'e Paris Diderot - Paris 7, France.}
\date{24 September, 2012}
\maketitle
\bibliographystyle{plain}

\begin{abstract}
We prove some results on the Wadge order  on the space of sets of natural numbers
endowed with Scott topology,
and more generally, on omega-continuous domains.
Using alternating decreasing  chains we characterize the 
property of Wadge hardness for the classes of the Hausdorff difference hierarchy
(iterated differences of open sets). 
A similar characterization holds for Wadge one-to-one  
and finite-to-one completeness. 
We consider the same questions  for the effectivization of the Wadge relation.
We also show that for the space of sets of natural numbers endowed with the Scott topology,
 in each class  of the Hausdorff difference hierarchy
there are two strictly increasing chains of Wadge degrees of
sets properly in that class.
The length of these chains is the rank of the considered class, and each
element in one chain is incomparable with all the elements in the other chain.
\end{abstract}

\section{Introduction}

Wadge reducibility between subsets of possibly uncountable
topological spaces is the topological version of the classical
 many-one reducibility between 
subsets of discrete countable spaces like words or integers.
Since Wadge reducibility defines a preorder on subsets of the topological space, 
one can consider the associated equivalence relation,
whose classes are called Wadge degrees.
The theory of Wadge degrees is  remarkable
in the case of Polish
(i.e. completely metrizable with some countable dense subset)
zero-dimensional
(i.e. admitting some basis of clopen sets)
topological spaces,
for instance the Cantor space $\cantor$
and the Baire space $\Baire$ with the product of discrete topologies.
William~W.~Wadge~(1972) developed such theory by looking 
at reductions between subsets of $\N^\omega$ as winning strategies
in adequate infinite games, the now so-called Wadge games.
The strength of Wadge's theory comes from Martin's  
result (1975) about determinacy of Borel games.
The Wadge order on Borel subsets of $\Baire$
is well-founded and almost total:
every antichain has only two elements, 
and it is constituted by the Wadge
degree of a set and that of its complement.
In spaces that are not  totally disconnected this beauty breaks down. 
For instance, in the real line $\R$,  the Wadge order  is not 
well founded~(Hertling 1996a,1996b),
for each  $2\leq\xi<\omega_1$
there are chains  of proper $\tbSigma^0_\xi$ sets indexed by any countable ordinal~\cite{ikegami2010}
and, letting $\subseteq^*$ denote inclusion up to a finite set,
the preorder $(\PN,\subseteq^*)$
(hence also any partial order of cardinality $\aleph_1$)
can be embedded in the Wadge preorder
of Borel subsets of $\R$~\cite{ikegamischlichttanaka2012}.
Also in every metric space which is positive-dimensional
(i.e. at some point there is no clopen local base), there are
$2^{\aleph_0}$ Wadge incomparable sets in the Hausdorff
class $\tbD_2$ of differences of two open sets~\cite{schlicht2012}.

Outside zero-dimensional Polish spaces,
Wadge theory has been first considered in the space $\PN$ of sets of natural numbers 
endowed with the Scott topology (which is $T_0$ but not $T_2$, i.e., not Hausdorff), 
by A. Tang (1979,1981).
Then it has been extensively studied by V. Selivanov in a series of papers
(Selivanov 2004, 2005a, 2005b, 2006, 2008 and references there)
  in the context of $\omega$-continuous domains.
Although much of Wadge theory  fails in spaces with Scott topology, 
a non trivial part remains. 
For instance, Selivanov  showed  the non existence of self-dual degrees
in the Hausdorff difference hierarchy.

Another problem is to understand what part of Wadge theory remains in effective topological spaces by considering effective reductions.   Since Borel determinacy gives highly 
non computable game strategies, the main tool used in classical 
Wadge  theory is not available.
As pointed out to us by an anonymous referee, 
it  follows from
Fokina, Friedman and Tornquist 2010 that 
the structure of effective Wadge
degrees is very different, even for the Baire space:
the ordered structure of inclusion on computably enumerable subsets of $\N$
can be embedded into the effective Wadge hierarchy of
$\Sigma^0_2$ subsets of the Baire space;
cf. Proposition~\ref{p:embedding inclusion in effective degrees} infra.

In this paper we concentrate on Wadge theory on the classes  the Hausdorff difference hierarchy, i.e. on $\tbDelta^0_2$ Borel sets.
We also consider the same problems but using only effective reductions. 
With this study we provide the general framework for
the particular cases of effective Wadge hardness
we used  to define highly random reals in (Becher and Grigorieff 2005, 2009).
The paper is organized as follows. 
\S\ref{s:preliminaries} recalls the basic 
definitions of
the Borel and Hausdorff hierarchies and some material on domains
and effective versions of various notions. 
In \S\ref{s:Wadge}  we present  Wadge reductions  in general topological spaces, 
and their effective versions.
We make explicit the relation between Wadge completeness and universality,
and we recall what is already known about Wadge theory outside zero-dimensional spaces.  
In \S\ref{s:chains} we prove the main results of the paper.
Theorem~\ref{thm:Dalphachains} characterizes
Wadge hardness for the classes of the Hausdorff difference hierarchy
in terms of alternating decreasing  chains.
Theorem~\ref{thm:one to one} gives a similar characterization
for Wadge one-to-one  and  finite-to-one completeness.  
We  prove effective versions of these results.

Theorem~\ref{thm:chain degrees} shows that,
for the space of sets of natural numbers endowed with the Scott topology,
 in each class  of the Hausdorff hierarchy there are two 
strictly increasing chains of Wadge degrees of sets properly in that class.
The length of these chains is the rank of the considered class, and each
element in one chain is incomparable with all the elements in the other chain.
This result builds on the construction done  by Selivanov (2005b)
of two incomparable Wadge degrees   
 which properly belong to an
arbitrarily given Hausdorff class.
Finally in \S\ref{s:PN versus} we prove  results on 
Wadge theory for the Borel class~$\tbSigma^0_2$.

\section{Preliminaries}\label{s:preliminaries}

For any set $A$, $\+P A$ denotes the powerset of $A$
and $\+P_{\!<\omega} A$ denotes the family of all finite subsets of $A$.
We write $\N$ for the set of natural numbers.
When $X$ is a subset of $\N$, and $i\in\N$,  $X+i=\{ n+i \ | \ n\in X\}$ and 
$iX=\{ i\ n \ | \ n\in X\}$.
Subsets of $\N$ are simply called  {\it sets},
subsets of $\PN$ are called {\it families},
and   collections of subsets of $\PN$ are called  {\it classes}.
As usual, $\Baire$ is the set of infinite sequences of natural numbers
 (the Baire space),
and $\Z, \Q, \R$ denote the sets of integer, rational and real numbers, respectively.
Greek letters $\alpha,\beta, \gamma, \delta$ are used to denote ordinals. 
We write 
$\omega$ for the first infinite ordinal, $\omega_1$ for the first uncountable ordinal,
and $\CK$ for the least non computable ordinal (the Church-Kleene ordinal).
For any two ordinals $\alpha,\beta$,  
$\alpha\sim\beta$ means that they have the same parity.

\subsection{Borel and Hausdorff Hierarchies}\label{ss:hierarchies}

Following (Selivanov 2005b, 2008),
we consider the Borel hierarchy and the Hausdorff-Kuratowski difference hierarchy in 
a general topological space $P$.
The Borel hierarchy consists of classes of subsets of $P$, namely
$\tbSigma^0_\alpha, \tbPi^0_\alpha, \tbDelta^0_\alpha$, where $1\leq\alpha<\omega_1$.
They are defined by induction on $\alpha$~:
$\tbSigma^0_1$ is the class of open subsets of $P$,
$\tbSigma^0_2$ is the class of countable unions of differences of open subsets of $P$,
and, for $\alpha>2$, $\tbSigma^0_\alpha$ is the class of countable unions of
sets in  $\bigcup_{\beta<\alpha}\tbPi^0_\beta$.
The class $\tbPi^0_\alpha$ is the class of complements
of sets in $\tbSigma^0_\alpha$ and
$\tbDelta^0_\alpha=\tbSigma^0_\alpha\cap\tbPi^0_\alpha$.
The class $\bGd$ (respectively $\bFs$) is the family of
countable intersections of open sets (respectively countable unions of closed sets).
In general, it is a proper subclass of $\tbPi^0_2$
(respectively $\tbSigma^0_2$).

For an ordinal $1\leq\alpha<\omega_1$, the operation $D_{\alpha}$
sends an $\alpha$-sequence of sets $(A_\beta)_{\beta<\alpha}$ to the set
$D_{\alpha}((A_\beta)_{\beta<\alpha})= 
\bigcup \{ A_{\beta}\setminus
\cup_{\gamma<\beta} A_\gamma \mid \beta<\alpha,\ \beta\not\sim\alpha\}$
(recall that $\beta\not\sim\alpha$ means that $\alpha,\beta$ have different parities).
The Hausdorff hierarchy is constituted by the Hausdorff classes
$\tbD_\alpha, \co\tbD_\alpha$ where $1\leq\alpha<\omega_1$.
The class $\tbD_\alpha$ consists of all subsets of $P$ of the form 
$D_\alpha((A_\beta)_{\beta<\alpha})$, 
where $(A_\beta)_{\beta<\alpha}$ is an $\alpha$-sequence of open subsets of $P$
(with no loss of generality, this $\alpha$-sequence can be supposed
increasing with respect to the inclusion relation).
The class $\co\tbD_\alpha$ is the class of complements of sets in $\tbD_\alpha$.
The Hausdorff-Kuratowski hierarchy is obtained by replacing
open sets by sets in a $\tbSigma^0_\xi$ class, $\xi<\omega_1$~:
its classes are denoted by
$\tbD_\alpha(\tbSigma^0_\xi)$ and $\co\tbD_\alpha(\tbSigma^0_\xi)$.
Hausdorff-Kuratowski theorem ensures that
$\bigcup_{1\leq\alpha<\omega_1}\tbD_\alpha(\tbSigma^0_\xi)=~\tbDelta^0_{\xi+1}$ for any $1\leq\xi<\omega_1$.

When the space $P$ is not obvious from the context, or when we want to make it explicit,
we will add a parenthesized reference to the space~$P$, writing 
$\tbSigma^0_\alpha(P)$, 
$\tbPi^0_\alpha(P)$,
$\tbDelta^0_\alpha(P)$,
$\bGd(P)$, 
$\bFs(P)$,
$\tbD_\alpha(P)$,
$\tbD_\alpha(\tbSigma^0_\xi(P))$,
$\co\tbD_\alpha(\tbSigma^0_\xi(P))$.
Let us mention the following classical result.

\begin{fact}\label{f:inverse image Borel}
If  $f:P\to Q$ is continuous then the inverse image of a subset of $Q$
in some Borel or Hausdorff-Kuratowski class of~$Q$
is in the corresponding class of subsets of~$P$.
\end{fact}

\subsection{Domains}\label{ss:domains}

We briefly recall the main definitions 
and refer the reader to classical papers and books,
for instance~\cite{abramsky1994,edalat1997,BookDomains}.
A {\it directed complete partial order (dcpo)} is a partially ordered set
$(P,\sqsubseteq)$ such that every non empty directed subset $S$ has a
least upper bound (denoted by $\sqcup S$).
A subset $X$ of $P$ is an upset (respectively downset) if, for all $x,y\in P$,
if $x\in X$ and $x\sqsubseteq y$ 
(respectively $y\sqsubseteq x$) then $y\in X$.
The {\it Scott topology} on a dcpo $P$ admits as closed sets
all downsets closed under suprema of directed subsets of $P$.
Thus, open subsets are upsets $O$ such that every directed set
with supremum in $O$ has an element in $O$.
The Scott topology is $T_0$
(if $x\neq y$ then there exists an open set which contains
only one of the two points $x,y$).
The order relation $\sqsubseteq$ can be recovered from the Scott topology
as the specialization order: $x\sqsubseteq y$ if and only
if every open set which contains $x$ also contains $y$.
A function $f:P\to Q$ between two dcpo's is {\em continuous} with respect
to the Scott topologies if and only if it is increasing
and preserves suprema of directed subsets: if $S\subseteq P$ is directed
then $f(\sqcup S)=\sqcup f(S)$.

The approximation (or way-below) relation 
in  a dcpo $(P,\sqsubseteq)$ is
defined as follows: $x\ll y$ if, for all directed subsets $S$,
$y \sqsubseteq \sqcup S$ implies $x\sqsubseteq s$ for some $s\in S$.
Thus, $x\ll y \Rightarrow x\sqsubseteq y$ and 
$x' \sqsubseteq x\ll y\sqsubseteq y'\ \Rightarrow\ x'\ll y'$.
An element $x\in P$ is {\it compact (or finite)} if $x\ll x$.
If $x$ is compact then $x\ll y \iff x\sqsubseteq y$.
The set of compact elements is denoted $K(P)$.
A {\it continuous domain} is a dcpo such that,
for every $x\in P$, the set $\ddown x =\{z\in P \mid z\ll x\}$
is directed and $x=\sqcup \ddown x$.
A basis of a continuous domain is any set $B$ such that, for all~$x$,
$B\cap\!\ddown x$ is directed and $x=\sqcup(B\cap\!\ddown x)$.
An {\it $\omega$-continuous domain} is a dcpo admitting a
countable basis.
An {\it $\omega$-algebraic domain} is a dcpo for which $K(P)$
is a countable basis.
If $B$ is a basis of a continuous domain then the sets 
$\uup b$,  for $ b\in B$, form  a topological basis,
where $\uup x=\{y\mid x\ll y\}$.
Let us mention the classical interpolation property.

\begin{fact}\label{f:interpolation}
Let $B$ be a basis of the continuous domain $P$.
Suppose $y\in P$  and $X\subseteq P$ is finite such that for each $x\in X$, $x\ll y$. Then, there exists $b\in B$ such that for each $x\in X$ \ $x\ll b$ and~$b\ll y$.
\end{fact}

\subsection{The Scott Domain $\PN$}\label{ss:PN}

The ordered set $(\PN,\subseteq)$ is an $\omega$-algebraic domain,
its compact elements are the finite sets
and $X\ll Y$ if and only if $X$ is finite and included in $Y$.
A topological basis of the Scott  topology is the class $\{\+B_A\mid A\in\FIN\}$
where $\+B_A=\{X\in\PN\mid A\subseteq X\}$.
The Scott topology gives ``positive information'' about sets,
and contrasts with the Cantor topology on $\cantor$ which gives positive
and negative information about sets via  their characteristic functions.
As recalled before, the Scott topology on $\PN$ is $T_0$:
if~$i\in X\setminus Y$, the open set $\+B_{\{i\}}$ contains $X$ but not~$Y$.

As observed by Selivanov (2005b)
the classes of finite rank of the Scott Borel hierarchy in $\PN$ do not coincide
with the corresponding ones in the Cantor space $\cantor$~: for all~$n\in\N$,
$n\geq1$,
$\tbSigma^0_n(\PN)
\subsetneq \tbSigma^0_n(\cantor)
\subsetneq \tbSigma^0_{n+1}(\PN)$.
For instance,
$\+X=\PN\setminus\{\N\}$, defined by the formula
$\exists x\ (x\notin X)$, is $\tbSigma^0_1(\cantor)$ and
$\tbSigma^0_2(\PN)$ but neither Scott open nor Scott closed.
However,  the classes of infinite rank of the Borel hierarchy 
in $\PN$ and $\cantor$ coincide.
The only subfamilies of $\PN$ which are both open and $\bFs$
are $\emptyset$ and $\PN$.
To see why, suppose $\+O$ is open, $\+O\neq\emptyset$,
and $\+X$ is $\bFs$,  $\+X\neq\PN$.
Let $\+X=\bigcup_{i\in\N}(\PN\setminus\+O_i)$
where $\+O_i$ is open. Observe that $\+O_i$ is non empty.
Choose finite sets $A$ and $B_i$, for $i\in\N$,  such that
$\+B_A\subseteq\+O$ and $\+B_{B_i}\subseteq\+O_i$.
Let $C=A\cup\bigcup_{i\in\N}B_i$. 
Then $C\in\+O$ since $C\supseteq A$ 
and $C\in\+O_i$ since $C\supseteq B_i$.
Thus, $C$ is in $\+O\setminus\+X$,
showing $\+O\neq\+X$.
A map $f:\PN\to\PN$ is Scott continuous if and only if
$f(X)=\bigcup_{A\subseteq X, A\in\FIN} f(A)$
for all $X$.
If $f:\PN\to \PN$ is a continuous bijection
then there exists a permutation $\theta:\N\to\N$ such that
for all $X$, $f(X)=\{\theta(x)\mid x\in X\}.$
In particular, $f^{-1}$ is also continuous.
Indeed, the inclusion order is a topological notion
(namely, the specialization order)
hence $f$ is an automorphism of $(\PN,\subseteq)$ and
as such, respects unions and maps singleton sets into singleton sets.
Let us mention a classical result which shows that $\PN$ is ``universal" for
$\omega$-continuous domains.

\begin{fact}\label{f:PN}
Let $(b_n)_{n\in\N}$ enumerate a countable basis
of the $\omega$-continuous domain $(P,\sqsubseteq)$.
Let $\varphi:P\to\PN$ be such that $\varphi(x)=\{n\mid b_n\ll x\}$.
\begin{enumerate}
\item
The map $\varphi$ is a dcpo embedding
(i.e. an isomorphism between
$(P,\sqsubseteq)$ and $(\varphi(P),\subseteq)$ such that
$\{x\in P\mid \varphi(x)\subseteq Z\}$ is directed for any
$Z\in\PN$).

\item
The map $\varphi$ is is a topological embedding
of $P$ into $\PN$ (i.e. an homeomorphism from $P$
to the subspace $\varphi(P)$ of $\PN$).
Then, for any Borel or Hausdorff-Kuratowski class $\tbGamma$,
$\{\varphi(Z)\mid Z\in\tbGamma(P)\}
=\tbGamma(\varphi(P))=\{\+X\cap\varphi(P) \mid \+X\in\tbGamma(\PN)\}$.

\item
$\varphi(P)$ is $\tbPi^0_2$ in $\PN$.
\end{enumerate}
\end{fact}

\begin{proof}
We just prove items 2 and 3.  Since $\varphi$ is injective,
it suffices to consider the case of basic open sets $\uup b_n$'s of $P$.
Observe that
$\varphi(\uup b_n)=\{\varphi(x) \mid x\in P,\ b_n\ll x\}
=\+B_{\{n\}}\cap\varphi(P)$. This proves item 2.
For item 3 observe that $X\in\varphi(P)$ if and only if,
relative to the approximation relation $\ll$ of $P$,
$\{b_i\mid i\in X\}$ is a directed downset. 
Letting $R=\{(i,j)\mid b_i\ll b_j\}$,
this means that $X$ satisfies the formulas
$\forall i,j\ \exists k\ (i,j\in X\Rightarrow
                    (k\in X\wedge(i,k)\in R\wedge(j,k)\in R))$
and
$\forall i,j\ (j\in X\wedge(i,j)\in R\Rightarrow i\in X)$,
i.e., 
$X\in\varphi(P)$ if and only if $X$ is in the $\tbPi^0_2$ families
$\bigcap_{i,j\in\N}\ 
\bigcup_{k\in Z_{i,j}}(\+B_{\{k\}}\cup(\PN\setminus \+B_{\{i,j\}}))$
and
$\bigcap_{(i,j)\in R}(\+B_{\{i\}}\cup(\PN\setminus \+B_{\{j\}}))$,
where   $Z_{i,j}=\{k\mid \ (i,k)\in R \mbox{ and }  (j,k)\in R\}$.
\end{proof}

\begin{remark}\label{rk:transfer PN pb}
Fact \ref{f:PN} allows us to transfer properties from $\PN$
to $\omega$-continuous domains.
\end{remark}

\subsection{Effectivization}\label{ss:effective}
\subsubsection{Presentations of topological spaces.}
The following definition slightly departs from (Weihrauch 2000, page 63).

\begin{definition}\label{def:effective space}
Let $P,Q$ be $T_0$ second-countable topological spaces.
\begin{enumerate}
\item 
A {\it presentation of the topology of $P$} is an enumeration 
(not necessarily injective)
$(U_n)_{n\in\N}$ of some topological basis of $P$.
For convenience, we shall sometimes consider enumerations indexed by $\N^2$ or $\FIN$.
\item
{\em A presentation $(U_n)_{n\in\N}$ is  effective}
if $\{(i,j) \mid U_i\subseteq U_j\}$ is computably enumerable.
An effective topological space consists of a
$T_0$ second-countable topological space $P$ together with an effective presentation of $P$.
\item
Relative to a presentation $(U_n)_{n\in\N}$ of the topology of $P$,
{\em an element $x$ of $P$ is effective} (respectively  {\it computable}) if the set
$\{i\in\N\mid x\in U_i\}$
 is computably enumerable (respectively computable).
Let $\alpha<\CK$ be an infinite computable ordinal.
A sequence $(x_\beta)_{\beta<\alpha}$ of elements of $P$ is {\it effective (respectively computable)}
if there exists a computable relation $R\subseteq\N^2$,
isomorphic to the ordering of $\alpha$, such that,
letting $\rho:\N\to\alpha$ be the isomorphism between
$(\N,R)$ and $(\alpha,<)$, the relation
$\{(n,i) \mid x_{\rho(n)}\in U_i \}$ is computably enumerable (respectively computable).
In that case, each element of the $\alpha$-sequence is effective (respectively computable).
\item
Relative to presentations $(U_n)_{n\in\N}$, $(V_n)_{n\in\N}$ of 
the topologies of $P$ and $Q$,
{\em a continuous map $f:P\to Q$ is  effective}
if the set $\{(i,j) \mid f(U_i)\subseteq V_j\}$ is computably enumerable.
Effective $\alpha$-sequences of maps are defined as above.
\end{enumerate}
\end{definition}

\begin{fact}
Effective maps are closed under composition.
They send effective elements onto effective elements,
idem with $\alpha$-sequences for $\alpha<\CK$.
\end{fact}

\begin{remark}\label{rk:product presentation}
If $(U_n)_{n\in\N}$, $(V_n)_{n\in\N}$ are presentations of the topological spaces
$P$ and $Q$, then $(U_i\times V_j)_{i,j\in\N}$
is a presentation of the product topological space $P\times Q$.
\end{remark}

\subsubsection{Effective Borel and Hausdorff hierarchies.}

Fixing a presentation $(O_n)_{n\in\N}$ of the topology of $P$,
one defines the effective versions of
the Borel and Hausdorff-Kuratowski hierarchies:
 the classes
$\Sigma^0_\alpha$, $\Pi^0_\alpha$, $\Delta^0_\alpha$,
$\tlD_\alpha$, $\co\tlD_\alpha$,
$\tlD_\alpha(\Sigma^0_\xi)$, $\co\tlD_\alpha(\Sigma^0_\xi)$,
for $1\leq\alpha,\xi<~\CK$.
The class $\Sigma^0_1$ consists of those open sets which are of the form
$\bigcup_{i\in I}O_i$ where $I$ is a c.e.
(computably enumerable) subset of $\N$.
The class $\Sigma^0_2$ consists of those sets of the form
$\bigcup_{n\in\N}(\bigcup_{i\in I_n}O_i)
\setminus(\bigcup_{j\in J_n}O_j)$
where $\{(n,i)\mid i\in I_n\}$ and $\{(n,j)\mid i\in J_n\}$ are c.e.
All Borel classes of finite rank are obtained in a similar way.
For the infinite ranks, the definition involves a more
complex machinery of Borel codes which can be developed in several ways,
 cf. (Moschovakis 1979/2009,~\S3H, 7B),
(Marker 2002,~\S7), (Selivanov 2008, \S3.2),
or (Becher and Grigorieff 2012, \S5.3 to 5.5).
The effective Hausdorff-Kuratowski classes are obtained similarly. 
Fact~\ref{f:inverse image Borel} has an effective version.

\begin{fact}\label{f:effective inverse image Borel}
If  $f:P\to Q$ is effective then the inverse image of a subset of~$Q$
in an effective Borel or Hausdorff-Kuratowski class of~$Q$
is in the corresponding effective class of~$P$.
\end{fact}

\subsubsection{Effective $\omega$-continuous domains.}
Definition~\ref{def:effective space} extends to domains. 
\begin{definition}\label{def:effective domain}
\begin{enumerate}
\item[1]
{\em A presentation of an $\omega$-continuous domain $P$}
is an enumeration of some basis of $P$.

\item[2]
{\em A presentation $(p_n)_{n\in\N}$ is  effective}
if $\{(i,j) \mid p_i\ll p_j\}$ is computably enumerable.
An effective $\omega$-continuous domain consists of an
$\omega$-continuous domain $P$ together with an effective presentation of $P$.
\end{enumerate}
Fix some presentations $(p_n)_{n\in\N}$, $(q_n)_{n\in\N}$ of
the (not necessarily effective) $\omega$-continuous domains $P$ and $Q$.
\begin{enumerate}
\item[3]
{\em An element $x\in P$ is domain-effective}
(respectively {\em domain-computable})
if the set $\{i\in\N\mid p_i\ll x\}$ 
is computably enumerable (respectively computable).
For $\alpha<\CK$, domain-effective (respectively domain-computable)
$\alpha$-sequences are defined as in Definition~\ref{def:effective space}
using the set $\{(n,i) \mid p_i\ll x_{\rho(n)}\}$.
Each element of a domain-effective (respectively domain-computable)
$\alpha$-sequence is domain-effective (respectively domain-computable).

\item[4]
{\em A continuous map $f:P\to Q$ is domain-effective} if
$\{(i,j) \mid q_j\ll f(p_i)\}$ is computably enumerable.
For $\alpha<\CK$, domain-effective $\alpha$-sequences of maps
are defined as above.
\end{enumerate}
\end{definition}

\begin{example}\label{ex:PN effective}
Consider the sequence $(\+B_A)_{A\in\FIN}$,
where $\+B_A=\{X\in\PN\mid A\subseteq X\}$.
It is an effective presentation of $\PN$,
(called the canonical presentation).
An element $X\in\PN$ is domain-effective 
(respectively domain-computable)
relative to this presentation
 if and only if it is
a computably enumerable (respectively computable) subset of $\N$.
A~map $f:\PN\to\PN$ is domain-effective
if and only if there exists a computable function $g:\FIN \to \N$
such that, for all $X$, 
$f(X)=\bigcup_{A\subseteq X, A\in\FIN} W_{g(A)}$
where $W_e$ denotes the computably enumerable subset of $\N$
with code~$e$.
\end{example}

\begin{remark}\label{rk:product domain presentation}
If $(p_n)_{n\in\N}$, $(q_n)_{n\in\N}$ are presentations of the
$\omega$-continuous domains $P$ and $Q$
then $(p_i,q_j)_{i,j\in\N}$ is a presentation of the product
$\omega$-continuous domain $P\times Q$.
\end{remark}

\begin{fact}\label{f:effective map and effective element}
Let $(p_n)_{n\in\N}$, $(q_n)_{n\in\N}$ be presentations of the
$\omega$-continuous domains $P,Q$.
A domain-effective  map $f:P\to Q$ between $\omega$-continuous domains
sends any domain-effective element of $P$ to 
a domain-effective element of $Q$.
Idem with domain-effective $\alpha$-sequences for $\alpha<\CK$.
\end{fact}

\subsubsection{Topological effectiveness versus domain effectiveness.}
While Definition~\ref{def:effective space} gives the notions of element, sequence and map that are topology-effective, Definition~\ref{def:effective domain} does it for the domain-effective notions. Since  $x\in\uup{p_i}\iff p_i\ll x$, for  elements and sequences the two definitions of effectivity coincide.

\begin{fact}\label{f:effective element domain and topo}
Let $(p_n)_{n\in\N}$ be a presentation of the
$\omega$-continuous domain $P$.
An element $x\in P$ is domain-effective (respectively domain-computable)
if and only if it is topology-effective (respectively topology-computable).
Idem with $\alpha$-sequences, for $1\leq\alpha<\CK$.
\end{fact}

In the case of  $\omega$-algebraic domains, 
Definition~\ref{def:effective space} and~\ref{def:effective domain}
 the notion of effectivity coincides not only for elements and sequences but also for  space presentations  and  maps.

\begin{fact}\label{f:effective algebraic domain}
Let  $(p_n)_{n\in\N}$ be a presentation of the $\omega$-algebraic domain $P$
consisting of compact elements.
Let $(q_n)_{n\in\N}$ be a presentation of the $\omega$-continuous domain $Q$.
\begin{enumerate}

\item
The presentation $(p_n)_{n\in\N}$ of the $\omega$-continuous domain
is effective
if and only if so is the presentation $(\uup{p_n})_{n\in\N}$ of
the Scott topology of $P$.

\item
A map $f:P\to Q$ is domain-effective
if and only if it is topology-effective.
\end{enumerate}
\end{fact}

\begin{proof}
Since $p_i\in\uup{p_i}$ we have $\uup{p_i}\subseteq\uup{p_j}\iff p_j\ll p_i$
and $f(\uup{p_i})\subseteq\uup{q_j} \iff q_j\ll f(p_i)$.
\end{proof}

Finally, topological effective maps are also  domain-effective.

\begin{proposition}\label{p:effective domain and topo}
Let $(p_n)_{n\in\N}$, $(q_n)_{n\in\N}$ be presentations of the
$\omega$-continuous domains $P,Q$, such that $(p_n)_{n\in\N}$ is domain-effective.
Then,  a  topology-effective map $f:P\to Q$ 
is also domain-effective.
\end{proposition}

\begin{proof}
Let $S=\{(i,j)\mid q_j\ll f(p_i)\}$ and
$T=\{(i,j)\mid f(\uup{p_i})\subseteq\ \uup{q_j}\}$
be the two computably enumerable sets ensuring 
the domain-effectiveness and the topology-effectiveness of $f$.
Inclusion $S\subseteq T$ is straightforward.
Though, in general, $S$ and $T$ are distinct, we show that
$S = \{(i,j)\mid\exists k\ (p_k\ll p_i \wedge (k,j)\in T)\}$.
In particular, if $T$ is computably enumerable
then so is $S$.
Implication $p_k\ll p_i\wedge(k,j)\in T\Rightarrow(i,j)\in S$ is trivial.
Conversely, if $(i,j)\in S$ then $q_j\ll f(p_i)=\sqcup_{i\in\N}f(p_{i_n})$
where $(p_{i_n})_{n\in\N}$ is some increasing chain included in $\ddown{p_i}$
with supremum $p_i$.
By interpolation (cf. Fact~\ref{f:interpolation}),
let $q_j\ll y\ll\sqcup_{i\in\N}f(p_{i_n})$.
By definition of $\ll$, there is some $n$ such that $y\sqsubseteq f(p_{i_n})$
hence $q_j\ll f(p_{i_n})$.
Thus, letting $k=i_n$, we have $p_k\ll p_i$ and $(k,j)\in S\subseteq T$.
\end{proof}

\subsubsection{Effective embedding in $\PN$.}
We give a version of  Fact~\ref{f:PN}
that effectivizes the  embedding of an $\omega$-continuous
domain $P$ in $\PN$.

\begin{fact}\label{f:PN effective}
Let $(p_n)_{n\in\N}$ be a presentation of the $\omega$-continuous
domain $(P,\sqsubseteq)$.
\begin{enumerate}
\item
An element $x\in P$ is effective (respectively computable) 
if and only if the set
$\varphi(x)=\{n\in\N \mid p_n\ll x\}$ is effective (respectively computable)
relative to the canonical presentation of $\PN$.
Idem with $\alpha$-sequences for $\alpha<\CK$.

\item
If the presentation $(\uup{p_n})_{n\in\N}$ is effective
then $\varphi$ is a topology-effective embedding of $(P,\sqsubseteq)$
in $(\PN, \subseteq)$
(relative to the canonical presentation of $\PN$).
As a consequence,
for every effective Borel or Hausdorff-Kuratowski class $\Gamma$,

\centerline{$\{\varphi(Z)\mid Z\in\Gamma(P)\}
=\Gamma(\varphi(P))=\{\+X\cap\varphi(P) \mid \+X\in\Gamma(\PN)\}$.}

\item
If the presentation $(p_n)_{n\in\N}$ is effective then $\varphi(P)$ is $\Pi^0_2$ in $\PN$.
\end{enumerate}
\end{fact}

\begin{proof}
2. For $A\in\FIN$,
$\uup{A}\ll\varphi(\uup{p_i})$
if and only if 
$\varphi(\uup{p_i})\subseteq\+B_A$
if and only if 
$\forall n\in A\ \forall k\ (p_i\ll p_k\Rightarrow n\in \varphi(p_k))$
if and only if 
$\forall n\in A\ \forall k\ (p_i\ll p_k\Rightarrow p_n\ll p_k)$
if and only if 
$\forall n\in A\ (\uup{p_i}\subseteq\uup{p_n})$.
The universal quantification is bounded by the finite set $A$, so 
this last formula gives the computable enumerability of the set
$\{(A,i)\mid\ \uup{A}\supseteq\varphi(\uup{p_i})\}$.

3. Observe that in the proof of Fact~\ref{f:PN} item 3,
the sets $R$ and $\{(i,j,k)\mid k\in Z_{i,j}\}$ are computably enumerable.
\end{proof}
\begin{note}
Item 1 in Fact~\ref{f:PN effective} does not need that the presentation
$(p_n)_{n\in\N}$ be effective.
\end{note}

\section{Wadge Theory and Its Effectivization}\label{s:Wadge}

\subsection{Wadge Hardness and Completeness}\label{ss:Wadge}
%
\begin{definition}\cite{wadge72}
\label{def:Wadge}
Let $P, Q$ be two topological spaces and $\+C\subseteq\P(P)$.
\begin{enumerate}
\item
$X\subseteq P$ is {\it Wadge reducible} to $Y\subseteq Q$,
written $X \leq_W Y$,
if $X=f^{-1}(Y)$ for some continuous $f:P\to Q$.
\item
 A set $Y\subseteq Q$ is {\it Wadge hard} for $\+C$
if  every $X\in \+C$ is Wadge reducible to $Y$.
\item
 In case $P=Q$, the set $Y\subseteq Q$ is 
{\it Wadge complete} for $\+C$
if it is Wadge hard for $\+C$ and it belongs to $\+C$.
\item 
Wadge reducibility, Wadge hardness and Wadge completeness
  are {\em finite-to-one} or {\em one-to-one}
when the associated   reductions are, respectively,
finite-to-one or one-to-one.
\item
 The preordering $\leq_W$ on subsets of $P$ induces an equivalence relation.
Its equivalence classes are called {\it Wadge degrees}.
\end{enumerate}
\end{definition}

\begin{definition}
Relative to fixed presentations $(U_n)_{n\in\N}$ and $(V_n)_{n\in\N}$
of the topological spaces $P,Q$,
{\em effective Wadge reducibility, hardness and completeness}
 are obtained  by requiring that only effective maps be considered
in Definition~\ref{def:Wadge}.
\end{definition}
One-to-one and finite-to-one Wadge reducibilities are not equivalent.

\begin{proposition}\label{prop:oneonemanyone}
In the Scott domain $\PN$, Wadge reducibility,
finite-to-one  and one-to-one reducibilities
are not equivalent. Idem with effective reducibilities.
\end{proposition}

\begin{proof}
Let 
$\+O_1=\bigcup_{n\in\N}\+B_{\{n\}}$,
$\+O_2=\bigcup_{n>0}\+B_{\{n\}}$,
$\+O_3=\bigcup_{n>0}\+B_{\{0,n\}}$.
The complements of these families respectively
contain one set (namely $\emptyset$),
two sets (namely $\emptyset$ and $\{0\}$),
and infinitely many sets
(namely $\emptyset$, $\{0\}$ and all subsets of
$\PN\setminus\{0\}$).
Thus, there can be no injective reduction of $\+O_2$ to $\+O_1$,
and no finite-to-one reduction of $\+O_3$ to both $\+O_1$
and $\+O_2$.
However, there is a finite-to-one effective reduction of $\+O_2$ to $\+O_1$,
namely, $f(X)=X\setminus\{0\}$.
Also, there is an effective reduction $g$ of $\+O_3$ to both $\+O_1$ and $\+O_2$:
$g(X)=X$, if $0\in X$ and there is $n>0$ such that  $n\in X$; otherwise
$g(X)=\emptyset$.
\end{proof}

\begin{remark}
Related to the above proposition, cf. Theorem~\ref{thm:11 and finite to 1}
and Corollary~\ref{cor:one to one}.
\end{remark}

For zero-dimensional Polish spaces, Wadge theory has beautiful properties.

\begin{theorem}\cite{wadge72} \label{thm:duality}
Let $P,Q$ be zero-dimensional Polish spaces.
\begin{enumerate}
\item (Duality theorem).
Let $A\subseteq P$, $B\subseteq Q$ be Borel sets.
Then $A\leq_W B\ \ \mbox{or}\ \ B\leq_W (P\setminus A)$.

\item (Hardness theorem)
For all $\xi<\omega_1$, the following conditions are equivalent:
\begin{itemize}
\item $A$ is not $\tbPi^0_\xi(P)$.

\item $A$ is Wadge hard for $\tbSigma^0_\xi(S)$
for some uncountable zero-dimensional Polish space~$S$.

\item $A$ is Wadge hard for $\tbSigma^0_\xi(S)$
for every uncountable zero-dimensional Polish space~$S$.
\end{itemize}
\end{enumerate}
\end{theorem}

\begin{remark}
The zero-dimensional hypothesis is necessary in Theorem~\ref{thm:duality}
(cf. Hertling 1996a,1996b; Ikegami 2010; Schlicht 2011).
For instance, 
consider  the closed interval $A=[0,+\infty[$ in the real line.
Its complement $ \R\setminus A = \ ]\!\!-\infty,0[$ is open.
Let $B=\{\alpha\in\Baire:\forall n\ \alpha(n)\not =0\}$,
so $B$ is not open.
Since Wadge reductions are continuous, necessarily
$B \not\leq_W \R\setminus A$  
(a reduction is impossible because the inverse image of an open set
by a continuous map must be open)
and $A\not\leq_W B$ 
(a reduction is impossible because every continuous map
from a connected space to a totally disconnected space is constant).
\end{remark}
Wadge Duality theorem does not hold in Scott spaces. 
This contradicts (Tang 1981, page 365 line 2)
 where the theorem is qualified as ``straightforward" for $\PN$.
The failure of the Duality theorem 
leads to the well known observation that 
Wadge games, the main tool in Wadge theory,
can not be used to investigate the Wadge hierarchy on spaces that
are not zero-dimensional and Polish.

\begin{proposition}
Wadge Duality theorem fails in the Scott domain $\PN$.
\end{proposition}
\begin{proof}
Let $\+X = \+B_{ \{ 0\} }$, $\+Y=\FIN$.
Since $\PN\setminus \+X$ is closed but $\+Y$ is not closed,
we have $\+Y\not\leq_W \PN\setminus \+X$ .
Let us see that $\+X\not\leq_W\+Y$.
Since $2 \N +1 \not\in \+X$ but  $\{0\}\cup \ 2 \N +1 \in \+X$,
any reduction $f:\PN\to\PN$ from $\+X$ to $\+Y$
should be such that 
$f( 2 \N +1) \not\in \+Y$ hence infinite and
$f(\{  0\} \cup  \ 2 \N +1) \in \+Y$ hence finite.
But this is a contradiction since, being continuous,
reductions are  increasing with respect to subset inclusion.
\end{proof}

Theorem~\ref{thm:duality} does not have an effective version
because determinacy is not  guaranteed in the effective world.
As pointed to us by an anonymous referee,
the arguments given by Fokina, Friedman and Tornquist (2010)
yield the following result. 

\begin{proposition}\label{p:embedding inclusion in effective degrees}
The ordered structure of inclusion on
computable (respectively computably enumerable) subsets of $\N$
can be embedded into the ordered structure of effective Wadge degrees
(i.e. degrees relative to computable reductions) of
$\Pi^0_1$ (respectively $\Sigma^0_2$) subsets of the Baire space.
\end{proposition}
\begin{proof}
We follow the arguments in the proof of Theorem 9 in Fokina, Friedman and Tornquist (2010).
Their Theorem 6  establishes that
there is a uniform sequence $(A_n)_{n\in\N}$
of nonempty $\Pi^0_1$ subsets of $\Baire$
such that, for each $n$, there is no hyperarithmetical function
$F:\Baire\to\Baire$ such that $F(A_n)\subseteq\bigcup_{m\neq n}A_m$.
We also use that the computable sequence $0^\omega$
is in no $A_n$, stated as a Remark after their Theorem 5.

{\it Case of computable subsets of $\N$.}
For $f\in\Baire$ and $k\in\N$, let $kf\in\Baire$ be
such that
$(kf)(0)=k$ and $(kf)(x)=f(x+1)$ for all $x\in\N$.
For $X\subseteq\N$ computable,
let $A^*_X$ be the $\Pi^0_1$ set $\bigcup_{n\in X}nA_n$.
Consider $F:\Baire\to\Baire$ such that, for $f\in\Baire$,
$F(f)=f$ if $f(0)\in X$ and $F(f)=0^\omega$ otherwise.
Then $F$ is computable and reduces $A^*_X$ to $A^*_Y$
for any $X\subseteq Y\subseteq\N$.
Suppose now that $X,Y$ are subsets of $\N$,
$X\not\subseteq Y$,
and $F:\Baire\to\Baire$ reduces $A^*_X$ to $A^*_Y$.
For $g\in\Baire$, let $g^-\in\Baire$ be obtained by removing
the first element of $g$, i.e. $g^-(x)=g(x+1)$ for all $x$.
Pick some $k\in X\setminus Y$ and let $G:\Baire\to\Baire$ be such that $G(f)=F(kf)^-$ for all $f\in\Baire$.
Since $k\in X$, $G(A_k)=F(kA_k)^-\subseteq (A^*_Y)^-
=\bigcup_{m\in Y}A_m$.
Since $k\notin Y$, the assumed property of the $A_n$'s ensures that
$G$ hence also the reduction $F$ cannot be hyperarithmetical.
Thus, there is no computable reduction of $A^*_X$ to $A^*_Y$.

{\it Case of computably enumerable subsets of $\N$.}
For $X\subseteq\N$ computably enumerable,
let $A^\dagger_X$ be the $\Sigma^0_2$ set
$0^{<\omega}1A^*_X$.
Let $(n_i)_{i\in\N}$ be a computable enumeration of $X$
and consider the map $F:\Baire\to\Baire$ such that
$F(0^p 1 n_i f)=0^{p+i}1 n_i f$ for all $p,i\in\N$, $f\in\Baire$,
and $F$ takes value $0^\omega$ elsewhere.
Then $F$ is computable and reduces $A^\dagger_X$ to $A^\dagger_Y$
for any $X\subseteq Y$.
As before,
$A^\dagger_X$ cannot be effectively reduced to $A^\dagger_Y$
in case~$X\not\subseteq Y$.
\end{proof}

\subsection{Universality}\label{ss:universality}
It turns out that universality is related to Wadge completeness,
cf. (Moschovakis 1979/2009, page 27)
or (Kechris 1995, page 85, proof of Theorem 14.2).

\begin{definition}
\label{def:universal}
Let $P, Q$ be topological spaces, $\tbGamma$ a Borel or Hausdorff-Kuratowski class.
\begin{enumerate}
\item  A set $U\in \tbGamma(P\times Q)$ is
{\it $P$-universal} for  the class $\tbGamma(Q)$
if $\tbGamma(Q) = \{U_p \mid p\in P\}$
where $U_p=\{q \mid (p,q)\in U\}$.
\item
 A set $S\in \tbGamma(Q)$ is
{\it strongly $P$-universal} for $\tbGamma(Q)$
if there exists a continuous map $\Phi:P\times Q \to Q$
such that
$\tbGamma(Q) = \{\Phi_p^{-1}(S)\mid p\in P\}$ with $\Phi_p(q)=\Phi(p,q)$.
If $\Phi$ is one-to-one in its second argument
then $S$ is said to be {\it one-to-one strongly $P$-universal}.
\item
Fix some presentations of $P,Q$.
If $\Gamma$ is an effective class,
we get corresponding  {\it effective notions} by requiring
equality $\Gamma(Q) = \{U_p \mid p\in P, \textit{ $p$ effective}\}$
in item 1
and equality
$\Gamma(Q) = \{\Phi_p^{-1}(S)\mid p\in P, \textit{ $p$ effective}\}$
with $\Phi$ an effective map in item 2.
\end{enumerate}
\end{definition}

\begin{remark}
\begin{enumerate}
\item If there is some strongly $P$-universal set $S$ for $\tbGamma(Q)$
then there is some $P$-universal set $U$ for $\tbGamma(Q)$~:
let $U=\Phi^{-1}(S)$,
for a witness $\Phi$ of the strong $P$-universality of $S$.
\item  Effective strong universality is a 
variant of
the notion of universal partial computable function 
in Blum's isomorphism theorem
(recall that a partial computable $\psi:\N\to\N$ is universal
if there is some computable $f:\N^2\to\N$ such that
for all $e,x\in\N$, $\{e\}(x)=\psi(f(e,x))$,
cf.  Rogers 1967, pages 54, 191).
\end{enumerate}
\end{remark}

Wadge completeness lies between
universality and strong universality.

\begin{proposition}\label{p:CQQ}
Let $\tbGamma$ be a Borel or Hausdorff-Kuratowski class.
Let $P, Q$ be topological spaces.
\begin{enumerate}
\item
If $S$ is (one-to-one) strongly $P$-universal for $\tbGamma(Q)$ then
$S$ is (one-to-one)  Wadge complete for $\tbGamma(Q)$.
An effective version holds relative to fixed presentations of $P,Q$.
\item
Let $C(Q,Q)$ be the set of continuous maps $Q\to Q$ endowed with any
topology at least as fine as the topology of pointwise convergence.
If $S$ is Wadge complete for $\tbGamma(Q)$
then $U=\{(f,q)\in C(Q,Q)\times Q \mid f(q)\in S\}$ is
$C(Q,Q)$-universal for $\tbGamma(Q)$.
An effective version holds relative to presentations of $Q$
and of a topology on $C(Q,Q)$ finer than that of pointwise convergence.
\end{enumerate}
\end{proposition}

\begin{proof}
1. Let $X$ be in $\tbGamma(Q)$. Since $S$ is strongly universal,
there exists $p\in P$ such that $X= \Phi_p^{-1}(S)$
where $\Phi_p:Q\to Q$  is the continuous map  $q\mapsto\Phi(p,q)$.
Notice that if 
$\Phi$ is injective in the second
argument, then  $\Phi_p$ is injective.

2. The map $a:(f,q)\mapsto f(q)$
is continuous $C(Q,Q)\times Q\to Q$, so
$a^{-1}$ preserves Borel and Hausdorff-Kuratowski classes.
Thus,
$U=a^{-1}(S) \in \tbGamma(C(Q,Q)\times Q)$.
Since $S\in\tbGamma(Q)$, the set $U_f$ is in $\tbGamma(Q)$ for all $f\in C(Q,Q)$.
Since $S$ is Wadge hard for $\tbGamma(Q)$,
every $X\in\tbGamma(Q)$ is of the form $X=f^{-1}(S)$
for some $f\in C(Q,Q)$.
Hence, $f^{-1}(S)=U_f$.
If $Q$ is not a countable set, the pointwise topology on $C(Q,Q)$
has no countable basis. For the effective version we need a finer topology:
If $(V_i)_{i\in\N}$ is a presentation of  $Q$,
a convenient topology on $C(Q,Q)$ has
a presentation $(\+V_A)_{A\in\+P_{<\omega}\N^2}$
where $\+V_A=\{f\in C(Q,Q)\mid\forall(i,j)\in A\ f(U_i)\subseteq U_j\}$.
\end{proof}
Wadge completeness coincides with strong universality
when this last condition is not~void.

\begin{fact}
If there is a strongly $P$-universal set for $\tbGamma(Q)$
then every Wadge complete set for $\tbGamma(Q)$
is strongly $P$-universal for $\tbGamma(Q)$.
One-to-one and effective versions also hold.
\end{fact}
\begin{proof}
Suppose $U$ is a strongly $P$-universal set for $\tbGamma(Q)$ and $V$ is a Wadge complete set for $\tbGamma(Q)$.
Any continuous reduction from $V$ to $U$ 
composed with  a witness of the strongly $P$-universality of  $U$
yields a witness of the strongly  $P$-universality of $V$.
\end{proof}
There are also cases in which universality grants Wadge completeness.

\begin{fact}\label{f:universalcomplete}
Suppose $h:P\times Q \to Q$ is a homeomorphism.
If $U$ is $P$-universal for $\tbGamma(Q)$
then $h(U)$ is one-to-one Wadge complete for $\tbGamma(Q)$.
 \end{fact}
By classical arguments  (Kechris, 1995)
Wadge hardness and $P$-universality
can be lifted with set-theoretical operations of complements,
(countable) unions, intersections and differences.
When $P$, $P^2$ and $P^\omega$ (endowed with the product topologies) are all homeomorphic (for example
$N^\omega$, $2^\omega$ or $\PN$),
one can lift Wadge completeness (respectively,
$P$-universality) for a class $\mathbf{D}_\alpha(\mathbf{\Sigma}_\xi(Q))$ 
to that for a class $\mathbf{D}_\beta(\mathbf{\Sigma}_\mu(Q))$
as long as either $\xi < \mu$ or else $\xi = \mu$ and $\alpha$ 
is less or equal to~$\beta$.

\subsection{Known Outside the Baire Space}

The following simple result shows that
$\PN$ is for $\omega$-continuous domains
what the Baire space $\Baire$ is for Polish spaces.
It can be seen as an extension of Facts~\ref{f:PN} and~\ref{f:PN effective}.

\begin{proposition}\label{p:Baire PN universal}
For every Polish space (respectively $\omega$-continuous domain) $Q$ 
there is a $\Baire$-universal (respectively $\PN$-universal) set
for $\tbSigma^0_1(Q)$.
Given presentations of the topological spaces $\Baire$ and $Q$
(respectively the $\omega$-continuous domains $\PN$ and $Q$),
this universal set is in the effective class relative to the associated presentation
of $\Baire\times Q$ (respectively $\PN\times Q$),
cf. Remarks~\ref{rk:product presentation}, \ref{rk:product domain presentation}.
\end{proposition}

\begin{proof}
We argue for the case of $\omega$-continuous domains.
Let $(b_i)_{i\in\N}$ be a presentation of $Q$.
Let 
$U=\{(X,x)\in\PN\times Q\mid \exists i\in X\   b_i\ll x\}$.
Observe that $U\in\tbSigma^0_1(\PN\times Q)$ since
$U=\bigcup_{i\in\N}\+B_{\{i\}}\times \uup~b_i$.
Clearly, $U$ is in the effective class $\Sigma^0_1(\PN\times Q)$.
Also, $U$ is $\PN$-universal for $\tbSigma^0_1(Q)$
since, for every $I\subseteq\N$,
the open set $O=\bigcup_{i\in I}\uup b_i$, of $Q$ is equal to $U_I$.
\end{proof}

Applying Fact~\ref{f:universalcomplete} 
and Proposition \ref{p:Baire PN universal},
one obtains the existence of Wadge complete sets for
domains such as $\PN$,
a result proved in \cite{selivanov2005}.

\begin{corollary}\cite{selivanov2005}
\label{cor:Wadge complete exist}
If an $\omega$-continuous domain $Q$ is homeomorphic to $\PN\times Q$
then it admits some one-to-one Wadge complete set for
each Borel or Hausdorff-Kuratowski class.
\end{corollary}

Tang (1979)  proved that a family $\+A\in\tbDelta^0_2(\PN)$
is in $\bigcup_{n<\omega}\tbD_n(\PN)$
if and only if
there is a finite bound on the length of 
$\+A$-alternating monotone chains.
Of course, for finite chains there is no distinction between increasing or decreasing chains.
Using alternating trees Selivanov (2005b, 2006) extended 
and refined this result for $\omega$-algebraic domains.
A slight variation of the proof goes through with
$\omega$-continuous domains, cf. (Selivanov 2008), \cite{paper1}.

\begin{proposition}(Selivanov 2005b, Proposition 6.4 $i$)
Let $P$ be an $\omega$-continuous domain and $n\in\N$.
Every set in $\tbDelta^0_2(P)$
but  not in $\co\mathbf{D}_n(P)$
is Wadge hard for $\mathbf{D}_n(P)$.
\end{proposition}

Polish spaces and $\omega$-continuous domains have been put
in an elegant unifying framework: the {\em quasi-Polish} spaces \cite{deBrecht2011}.
Up to homeomorphism, quasi-Polish spaces coincide with
the $\tbPi^0_2$ subspaces of $\PN$.  
Wadge hardness results transfer from $\PN$ to quasi-Polish spaces as follows.

\begin{proposition}
Let $\tbGamma$ is $\tbSigma^0_\alpha(\PN)$ with $\alpha\geq3$
or $\tbPi^0_\alpha(\PN)$ with $\alpha\geq2$. 
Let $P$ be a topological space.
Any set $H\subseteq P$ that is Wadge hard for $\tbGamma(\PN)$
is also Wadge hard for  $\tbGamma(Q)$ for every quasi-Polish space $Q$.
\end{proposition}
\begin{proof}
Since every quasi-Polish space is homeomorphic to a $\tbPi^0_2$ subspace of $\PN$,
we can suppose $Q$ is a $\tbPi^0_2$ subspace of $\PN$. 
The hypothesis on $\tbGamma$ ensures that $Q$ is in $\tbGamma(\PN)$.
If $A \in \tbGamma(Q)$ then $A = A' \cap Q$ for some $A' \in \tbGamma(\PN)$
hence $A\in\tbGamma(\PN)$. 
The restriction to $Q$ of any continuous reduction of $A'$ to $H$
is a reduction of $A$ to $H$. 
\end{proof}

\section{Wadge Hardness and Alternating Decreasing Chains}
\label{s:chains}

\subsection{Alternating Chains}

Recall that we write 
 $\alpha\sim \beta$ to indicate that $\alpha$ and $\beta$ have the same  parity.

\begin{definition}
\label{def:chain}
Let $(P,\sqsubseteq)$ be a dcpo, $A$ a subset of $P$
and $\alpha$ an ordinal.
\begin{enumerate}

\item
An $(\alpha+1)$-sequence $(x_\beta)_{\beta\leq\alpha}$ of elements of $P$
is an {\it $A$-alternating decreasing chain} if
\begin{itemize}

\item
$x_\beta \sqsupset x_\delta$ for all $\beta<\delta\leq\alpha$, and

\item
$x_\beta\in A$ if and only if $\beta\not\sim \alpha$
(thus, $x_\alpha\notin A$).
\end{itemize}

\item
An {\it $A$-special $(\alpha+1)$-chain} is an $A$-alternating decreasing chain
$(x_\beta)_{\beta\leq\alpha}$ such that
\begin{itemize}

\item ($x_\beta\sqsupseteq x\sqsupset x_{\beta+1}) \Rightarrow (x\in A\Leftrightarrow\beta\not\sim\alpha)$,  and 

\item ($x_\alpha\sqsupseteq x) \Rightarrow (x\notin A)$
\end{itemize}
\end{enumerate}
\end{definition}

\begin{note} \label{note:alternating}
If $(x_\beta)_{\beta\leq\alpha}$ is $A$-alternating and $\gamma<\alpha$
then $(x_\beta)_{\beta\leq\gamma}$ is $A$-alternating if and only if $\gamma\sim\alpha$.
\end{note}

\begin{remark}\label{rk:no increasing alternating}
As proved by Selivanov (2005b, 2008),
when $\alpha$ is infinite and $A$ is $\tbDelta^0_2$
$A$-alternating increasing chains  do not exist.
In fact, if $A$ is in the difference of two open sets and
the supremum of an increasing chain $(x_n)_{n\in\N}$ is in $A$
then the $x_n$'s are in $A$ for all $n$ large enough.
By considering  countable unions, it follows that the same is true if $A$ is $\tbSigma^0_2$.
This property forbids $A$-alternation of increasing infinite chains if $A$ is $\tbDelta^0_2$.
\end{remark}

\subsection{Wadge Hardness and  Alternating Decreasing $(\alpha+1)$-Chains}

If $A$ is in $\tbDelta^0_2$
then long $A$-alternating  decreasing infinite chains do exist in $\PN$.

\begin{proposition}\label{p:decreasing chains PN}
Let $\alpha\geq1$ be a countable ordinal.
\begin{enumerate}
\item
There exists a family $\+A\in\tbD_\alpha(\PN)$ and an
$\+A$-special chain $(X_\beta)_{\beta\leq\alpha}$ such that
$X_\alpha$ and $X_\beta\setminus X_{\beta+1}$, for all $\beta<\alpha$,
are infinite subsets of $\N$.

\item
If $\alpha<\CK$ then $\+A$ can be taken in the effective class $\tlD_\alpha(\PN)$
and the chain can be taken computable
(cf. Definition~\ref{def:effective domain}). 
\end{enumerate}
\end{proposition}

\begin{proof}
Let $\beta\mapsto a_\beta$ be a bijection $\alpha\to3\N$.
Let $(A_\beta)_{\beta\leq\alpha}$ be a family of
pairwise disjoint infinite subsets of $3\N+1$.
For $\beta\leq\alpha$, let
$X_\beta =(3\N+2)\cup\bigcup_{\beta\leq\delta<\alpha}( \{a_\delta\}\cup A_\delta)$.
Then $X_\beta\setminus X_{\beta+1}=\{a_\beta\}\cup A_\beta$ and
$X_\alpha=3\N+2$ are infinite sets.
For $\beta<\alpha$, define an open family
$U_\beta = \{Z \mid \exists \gamma\leq\beta\ a_\gamma\in Z\}$.
With these open families, define a family
$\+A=D_\alpha((U_\beta)_{\beta<\alpha})
=\bigcup_{ \beta<\alpha,\beta\not\sim\alpha}
             U_\beta \setminus \bigcup_{\gamma<\beta} U_\gamma$
             in $\tbD_\alpha(\PN)$.
Observe that a set $X$ is in $\+A$ if and only if
it meets $\{a_\beta\mid\beta<\alpha\}$ and the least $\beta<\alpha$
such that $a_\beta\in X$ has parity different from that of $\alpha$.
In particular, if $X_\beta\supseteq X\supset X_{\beta+1}$ then $X\in\+A$ if and only if
$\beta\not\sim\alpha$.
Thus $(X_\beta)_{\beta\leq\alpha}$ is an $A$-special chain.
Effectivization is straightforward.
\end{proof}

In general in continuous domains, long decreasing chains may not exist.
For instance in the $\omega$-algebraic domain $(\mix,\sqsubseteq)$
of finite and infinite binary words with the prefix ordering,
every decreasing chain is finite.
However, in case a long decreasing chain exists then it can always be viewed as
an $A$-alternating chain for some $A\in\tbDelta^0_2$.

\begin{proposition} \label{p:decreasing implies A special}
Let $P$ be a continuous domain and $\alpha<\omega_1$ be an ordinal.
\begin{enumerate}
\item
Every strictly decreasing chain $(x_\beta)_{\beta\leq\alpha}$ in $P$
is $A$-special for some $A \in \mathbf{D}_\alpha(P)$.

\item
Suppose $P$ is $\omega$-continuous and fix some presentation of $P$.
If $\alpha<\CK$ and the chain is computable
then we can take $A\in\tlD_\alpha(P)$.
\end{enumerate}
\end{proposition}

\begin{proof}
Let $B$ be a basis of $P$. In the vein of Fact~\ref{f:PN},
we  use the isomorphism $\psi$ between $(P,\sqsubseteq)$
and a subspace of $(\+P(B),\subseteq)$ such that
$\psi(x)=B\cap\ddown x= \{b\in B\mid b\ll x\}$.
Thus, $(\psi(x_\beta))_{\beta\leq\alpha}$ is a strictly decreasing chain.
Let $U_\beta=\bigcup\{\uup b\mid b\in \psi(x_0)\setminus \psi(x_{\beta+1})\}$
for $\beta<\alpha$.
The $U_\beta$'s are strictly increasing open subsets of $P$.
Let $A=D_\alpha((U_\beta)_{\beta<\alpha})$.
The set $A$ is in $\tbD_\alpha(P)$.
Suppose $\beta<\alpha$ and $x_\beta\sqsupseteq x\sqsupset x_{\beta+1}$.
Then $\psi(x_\beta)\supseteq\psi(x)\supset\psi(x_{\beta+1})$, so that
$x\in U_\delta$ if and only if $\psi(x)$ meets $\psi(x_0)\setminus \psi(x_{\delta+1})$
                          if and only if $\delta\geq\beta$.
In particular, $x\in U_\beta\setminus\bigcup_{\gamma<\beta}$.
Thus, $x\in A$ if and only if $\beta\not\sim\alpha$ (hence, $\beta\sim\alpha+1$).
Finally, suppose $x_\alpha\sqsupseteq x$.
Then $\psi(x_\alpha)\supseteq\psi(x)$ and $x$ is in no $U_\beta$ hence $x\notin A$.
For the effective version, using a presentation $(b_i)_{i\in\N}$,
replace $\psi$ with 
$\varphi$ of Facts~\ref{f:PN} and \ref{f:PN effective}
such that $\varphi(x)=\{i\in\N\mid b_i\ll x\}$.
The computability of the chain $(x_\beta)_{\beta\leq\alpha}$
ensures that of 
$(\varphi(x_0)\setminus \varphi(x_{\beta+1}))_{\beta\leq\alpha}$.
\end{proof}

\begin{note}\label{note:alpha chains}
An effective chain is not enough for the above proof:
being differences of computably enumerable sets,
the sets $\varphi(x_0)\setminus \varphi(x_{\beta+1})$
might not be computably enumerable.
\end{note}

\begin{remark}\label{rk:alpha chains}
The proof of Proposition~\ref{p:decreasing implies A special}
amounts to a proof for the case $P=\+P(\kappa)$, where $\kappa$ is cardinal,
and a transfer to all continuous domains having a basis of cardinality $\leq\kappa$
via an obvious extension of Fact~\ref{f:PN}.
Idem for effectivization with $\kappa=\omega$ and Fact~\ref{f:PN effective}.
\end{remark}

We present now one of the main theorems of the paper.
For  the effective part of the result we use the following convention:
effective reductions $P\to Q$ are relative to the presentations
$(\uup{p_n})_{n\in\N}$ and $(\uup{q_n})_{n\in\N}$ of the Scott topologies
on $P,Q$ associated to presentations $(p_n)_{n\in\N}$ and $(q_n)_{n\in\N}$
of the $\omega$-continuous domains $P,Q$.

\begin{theorem}\label{thm:Dalphachains}
Let $Q$ be a continuous domain, $H\subseteq Q$ and
$\alpha<\omega_1$ be an ordinal.
\begin{enumerate}
\item
The following conditions are equivalent.
\begin{enumerate}
\item[$(i)$]\ \
$H$ is Wadge hard for $\tbD_\alpha(P)$ for every continuous domain $P$.

\item[$(ii)$]\ \
$H$ is Wadge hard for $\tbD_\alpha(\PN)$.

\item[$(iii)$]\ \
$H$ is Wadge hard for $\tbD_\alpha(P)$ for some continuous domain $P$
admitting a strictly decreasing chain of length $\alpha+1$.

\item[$(iv)$]\ \
There exists a  decreasing $H$-alternating chain in $Q$ of length $\alpha+1$.
\end{enumerate}

\item
Suppose $Q$ is an $\omega$-continuous domain and fix some presentation
$(q_n)_{n\in\N}$ of $Q$.
If $\alpha<\CK$ the following conditions are equivalent.
\begin{enumerate}
\item[$(i_e)$]\ \
$H$ is Wadge hard for $\tbD_\alpha(P)$
and effectively Wadge hard for $\tlD_\alpha(P)$
for every effective $\omega$-continuous domain $P$
(cf. Definition~\ref{def:effective domain}).

\item[$(ii_e)$]\ \
$H$ is effectively Wadge hard for $\tlD_\alpha(\PN)$.

\item[$(iii_e)$]\ \ \
$H$ is effectively Wadge hard for $\tlD_\alpha(P)$
for some effective $\omega$-algebraic domain $P$
admitting a computable strictly decreasing chain of length $\alpha+1$.

\item[$(iv_e)$]\quad
There exists an effective  decreasing
$H$-alternating chain in $Q$ of length $\alpha+1$.
\end{enumerate}
\end{enumerate}
\end{theorem}

\begin{proof}
$(i) \Rightarrow (ii)$.
Particularize item $(i)$ with $P=\PN$.

$(ii)\Rightarrow (iii)$.
Let $P=\PN$ and use Proposition~\ref{p:decreasing chains PN}.

$(iii) \Rightarrow (iv)$.
By Proposition~\ref{p:decreasing implies A special},
the strictly decreasing chain in $P$ of length $\alpha+1$
given by $(iii)$ is $A$-alternating for some $A\in\tbD_\alpha(P)$.
A continuous reduction of $A$ to $H$ maps this $A$-alternating chain
onto a decreasing $H$-alternating chain in $Q$.

$(iv) \Rightarrow (i)$.
Let $(y_\beta)_{\beta\leq\alpha}$ be a
decreasing $H$-alternating chain of length $\alpha+1$.
Suppose $X$ is a set in $\tbD_\alpha(P)$.
Then $X=\bigcup_{\beta<\alpha,\beta\not \sim \alpha }
             V_\beta \setminus \bigcup_{\gamma<\beta} V_\gamma$
for some increasing chain $(V_\beta)_{\beta\leq\alpha}$ of open sets.
Define $\tau:P\to \alpha+1$ and
$f:P\to Q$ as follows: for $z\in P$,
\smallskip

\quad $\tau(z)=\left\{\begin{array}{ll}
\textit{least $\beta<\alpha$ such that $z\in V_\beta$}
&\textit{if }z\in\bigcup_{\beta<\alpha}V_\beta
\\
\alpha&\textit{otherwise.}
\end{array}\right.$
\smallskip

\quad$f(z)=y_{\tau(z)}$
\smallskip
\\
Since the $V_\beta$'s are increasing and are open hence are upsets in $P$,
$\tau$ is  decreasing and $f$ is  increasing.
Letting $B$ be a basis of $P$,
we have $V_\beta=\bigcup_{c\in I_\beta}\uup c$
for some $I_\beta\subseteq B$.
Thus,
\smallskip
\\
$\begin{array}{lcl}
z\in V_\beta &\Leftrightarrow& \exists c\in I_\beta\ c\ll z\\
                    &\Leftrightarrow& \exists c\in I_\beta\ \exists b\in B\ \ c\ll b\ll z
\text{\hspace{1cm}(by the interpolation property)}\\
                    &\Leftrightarrow& \exists b\in B\ (b\ll z \ \wedge\ b\in V_\beta)
\\
\tau(z)&=&\min\{\tau(b)\mid b\in B,\ b\ll z\}
\\
f(z)&=&\max_{\sqsubseteq}\{f(b)\mid b\in B,\ b\ll z\}.
\end{array}$
\smallskip\\
The last equality shows that $f$ is continuous.
Finally, observe that $z\in X$ if and only if
$\tau(z)<\alpha\textit{ and }\tau(z)\not\sim\alpha$
if and only if $f(z)=y_{\tau(z)}\in H$.
Thus, $f$ reduces $X$ to~$H$.

$(i_e) \Rightarrow (ii_e)  \Rightarrow (iii_e)$.
Idem as $(i) \Rightarrow (ii)  \Rightarrow (iii)$.

$(iii_e) \Rightarrow (iv_e)$.
By Proposition~\ref{p:decreasing implies A special},
the chain in $P$ of length $\alpha+1$ given by $(iii)$
is $A$-alternating for some $A\in\tlD_\alpha(P)$.
Observe that an effective reduction from $A$ to $H$
maps this computable chain
onto an effective decreasing $H$-alternating chain.

$(iv_e) \Rightarrow (i_e)$.
Suppose $X\in\tlD_\alpha(P)$.
Keeping the notation as in the proof of $(iv)\Rightarrow(i)$,
we first show that the reduction $f:P\to Q$ of $X$ to $H$ is 
domain-effective.
Let $(p_i)_{i\in\N}$ and $(q_\ell)_{\ell\in\N}$ be presentations of $P$ and $Q$
such that $(p_i)_{i\in\N}$ is effective, i.e.
$\{(i,j)\mid p_i\ll p_j\}$ is computably enumerable.
Since $\alpha<\CK$ and the chain $(y_\beta)_{\beta\leq\alpha}$ is effective,
there exists an initial segment $S$ of $\N$,
a set $R\subset S^2$
and a map $\rho:S\to\alpha+1$ such that
\begin{itemize}
\item
$\rho$ is an isomorphism (necessarily unique) between $(S,R)$ and $(\alpha+1,\leq)$,
\item
the relation $\{(n,\ell)\mid q_\ell\ll y_{\rho(n)}\}$ is computably enumerable.
\end{itemize}
Since $X\in\tlD_\alpha(P)$, we can suppose that
$\{(n,j)\mid p_j\in I_{\rho(n)}\}$ is a computably enumerable set.
Let $a\in S$ be such that $\alpha=\rho(a)$.
For $\ell,i\in\N$,
\smallskip
\\
$\begin{array}{lcl}
q_\ell\ll f(p_i)&\iff& q_\ell\ll y_{\tau(p_i)}
\\
&\iff&q_\ell\ll \max\left(\{y_\beta\mid \tau(p_i)\leq\beta\leq\alpha\}\right)
\\
&\iff&q_\ell\ll y_\alpha \textit{ or }
\exists \beta<\alpha\ \left(q_\ell\ll y_\beta\textit{ and }\beta\geq\tau(p_i)\right)
\\
&\iff&q_\ell\ll y_\alpha \textit{ or }
\exists \beta<\alpha\ \left(q_\ell\ll y_\beta\textit{ and }p_i\in V_\beta\right)
\\
&\iff&q_\ell\ll y_\alpha \textit{ or }
\exists \beta<\alpha\ \left(q_\ell\ll y_\beta\textit{ and }\exists p_j\in I_\beta\ p_j\ll p_i\right)
\\
&\iff&q_\ell\ll y_{\rho(a)} \textit{ or }
\exists n\neq a\ \exists j\ \left(q_\ell\ll y_{\rho(n)}\textit{ and } p_j\in I_{\rho(n)}\textit{ and }p_j\ll p_i\right).
\end{array}$
\smallskip\\
Thus, the set $\{(i,\ell)\mid q_\ell\ll f(p_i)\}$ is obtained via conjunction, disjunction
and projection of computably enumerable sets.
As such, it is computably enumerable.
This proves  $f$ is a domain-effective map.

If $P$ is $\omega$-algebraic and the $p_n$'s are compact elements then 
Fact~\ref{f:effective algebraic domain} 
ensures that the reduction $f$ is also topology-effective.
In the general case where $P$ is only $\omega$-continuous, we argue as follows.
Consider the map $\varphi:P\to\PN$ such that
$\varphi(x)=\{n\in\N\mid p_n\ll x\}$.
Since the presentation $(p_n)_{n\in\N}$ of $P$ is effective,
Fact~\ref{f:PN effective} 
ensures that $\varphi$ is a topological
embedding which is topology-effective
(relatively to the canonical presentation of $\PN$) and that
$\varphi(X)=\+Y\cap\varphi(P)$ for some $\+Y\in\tlD_\alpha(\PN)$.
Since $\PN$ is an effective $\omega$-algebraic domain,
the above argument applied to the subset $\+Y$ of $\PN$,
yields a reduction $g:\PN\to Q$ of $\+Y$ to $H$ which is
topology-effective. 
The composition $g\circ\varphi:P\to Q$ is then a reduction of $X$ to $H$
which is also topology-effective.
\end{proof}

\begin{remark}
Letting $\alpha=1$  the characterization of hardness 
for $\tbSigma^0_1(\PN)$ given by Theorem \ref{thm:Dalphachains}
is in the vein of the work done in \cite{jsl2}.
\end{remark}
The next Corollary was known for the domain $\PN$
(Selivanov 2005b), cf. Corollary~\ref{cor:Wadge complete exist}.

\begin{corollary}\label{cor:complete D_alpha chains}
Let $Q$ be a continuous domain and $\alpha<\omega_1$ be a countable ordinal.
\begin{enumerate}
\item
If $Q$ admits some strictly decreasing chain of length $\alpha+1$
then there exists a Wadge complete set for $\tbD_\alpha(Q)$.
In particular, this applies if $Q=\PN$.

\item
Suppose $Q$ is an effective $\omega$-continuous domain.
If $Q$ admits some computable strictly decreasing chain of length $\alpha+1$
then there exists a Wadge complete set for $\tbD_\alpha(Q)$
which is also effectively Wadge complete for $\tlD_\alpha(Q)$.
In particular, this applies if~$Q=~\PN$.
\end{enumerate}
\end{corollary}

\begin{proof}
By Proposition~\ref{p:decreasing implies A special},
the assumed strictly decreasing chain of length $\alpha+1$
 is $H$-special, hence $H$-alternating for some $H\in\tbD_\alpha(Q)$.
Thus item (iv) of Theorem~\ref{thm:Dalphachains}  is satisfied; then, 
item (i) of Theorem~\ref{thm:Dalphachains}  also holds, hence  
$H$ is as desired.  The possible instantiation $Q=\PN$ comes from
Proposition~\ref{p:decreasing chains PN}.
Effectivization  is straightforward.
\end{proof}

\subsection{One-to-one Wadge Completeness and Alternating Decreasing Chains}

Universal properties as those in part 1 $(i)$ of Theorem~\ref{thm:Dalphachains} cannot hold with one-to-one Wadge reductions $f:P\to Q$ since such reductions can only exist  if $Q$ has cardinality at least that of $P$. Therefore, we shall restrict
 to $\omega$-continuous domains.
The version of Theorem~\ref{thm:Dalphachains} for one-to-one hardness
relies on the following notion.

\begin{definition}\label{def:scattered}
Let $(P,\sqsubseteq)$ be an ordered set and $\alpha$ an ordinal.
\begin{enumerate}
\item {\em A strictly decreasig chain  $(x_\beta)_{\beta\leq\alpha}$ is scattered} if
there exist one-to-one continuous maps $\theta_\beta:\PN\to P$,
for $\beta\leq\alpha$, such that
\begin{enumerate}

\item[]
$\theta_\alpha(\PN)\subseteq\{x\mid x_\alpha\sqsupseteq x\}$,
and 
$\theta_\beta(\PN)\subseteq\{x\mid x_\beta\sqsupseteq x\sqsupset x_{\beta+1}\}$
for $\beta<\alpha$.
\end{enumerate}

\item Let $A\subseteq P$.
{\em A strictly decreasig chain  $(x_\beta)_{\beta\leq\alpha}$ is 
 $A$-scattered} if it is scattered for one-to-one continuous maps 
$\theta_\beta:\PN\to P$,
for $\beta\leq\alpha$, such that
\begin{enumerate}
\item[]
$\theta_\beta(\PN)\subseteq A$ if $\beta\not\sim\alpha$, and 
$\theta_\beta(\PN)\subseteq P\setminus A$ if $\beta\sim\alpha$.
\end{enumerate}
\item  Suppose $(P,\sqsubseteq)$ is an $\omega$-continuous domain.
Fix some presentation of $P$.
If $\alpha<\CK$, a scattered chain $(x_\beta)_{\beta\leq\alpha}$ is effective
if it is an effective  $(\alpha+1)$-sequence of elements of $P$
and $(\theta_\beta)_{\beta\leq\alpha}$ is 
an effective $(\alpha+1)$-sequence of maps $\PN\to P$. 
\end{enumerate}
\end{definition}

\begin{proposition}\label{p:PN A scattered}
Let $(X_\beta)_{\beta\leq\alpha}$ be a strictly decreasing sequence
of subsets of~$\N$, of length $(\alpha+1)$.
\begin{enumerate}
\item  The following conditions are equivalent.
\begin{enumerate}
\item The chain $(X_\beta)_{\beta\leq\alpha}$ is scattered.
\item  $X_\alpha$ and each of  $X_\beta\setminus X_{\beta+1}$,   for $\beta<\alpha$,  is an infinite set.
\end{enumerate}

\item  Suppose $(X_\beta)_{\beta\leq\alpha}$ is $\+A$-special for some 
$\+A\subseteq\PN$.
Then,  a third equivalent condition~is
\begin{enumerate}
\item[(c)] $(X_\beta)_{\beta\leq\alpha}$ is $\+A$-scattered.
\end{enumerate}
\end{enumerate}
\end{proposition}

\begin{proof}
$(a)\Rightarrow(b)$. Suppose the chain is scattered.
Then, with the notation of Definition~\ref{def:scattered},
for $\beta<\alpha$,
$(\theta_\beta(\{0,\ldots,n\}))_{n\in\N}$ is a strictly increasing sequence of sets
in $\{X\mid X_\beta\supseteq X\supset X_{\beta+1}\}$.
Thus, $X_\beta\setminus X_{\beta+1}$ is infinite.
Similarly,  $X_\alpha$ is infinite. 

$(b)\Rightarrow(a)$.
Consider bijective maps $\mu_\alpha:\N\to X_\alpha$ and
$\mu_\beta:\N\to X_\beta\setminus X_{\beta+1}$, for $\beta<\alpha$.
For $Z\in\PN$, let 
$\theta_\alpha(Z)=\mu_\alpha(Z)$ and,
for $\beta<\alpha$, $\theta_\beta(Z)=X_{\beta+1}\cup\mu_\beta(Z)$.
The condition  of Definition~\ref{def:scattered} item 1  is true, hence the chain is scattered.

In case the chain is $\+A$-special, $(b)\Rightarrow (c)$
and $(c)\Rightarrow (a)$ are straightforward.
\end{proof}

\begin{corollary}\label{coro:PN A scattered}
Let $\alpha<\omega_1$ be an ordinal.
There exists $\+A\in\tbD_\alpha(\PN)$ and an $\+A$-scattered
chain of length $\alpha+1$.
If $\alpha<\CK$ then we can take $\+A\in\tlD_\alpha(\PN)$ and get
a computable $\+A$-scattered chain of length $\alpha+1$.
\end{corollary}

\begin{proof}
Use Propositions~\ref{p:decreasing chains PN} and \ref{p:PN A scattered}.
\end{proof}

To characterize  effective  Wadge  one-to-one  hardness we  use the same convention as in Theorem~\ref{thm:Dalphachains}:
effective reductions $P\to Q$ are  relative to 
 the presentations
$(\uup{p_n})_{n\in\N}$ and $(\uup{q_n})_{n\in\N}$ of the Scott topologies on $P,Q$ 
for presentations $(p_n)_{n\in\N}$ and $(q_n)_{n\in\N}$
of the $\omega$-continuous domains $P,Q$.

\begin{theorem}\label{thm:one to one}
Let $Q$ be an $\omega$-continuous domain, $H\subseteq Q$ and
$\alpha<\omega_1$ be an ordinal.
\begin{enumerate}
\item
The following conditions are equivalent.
\begin{enumerate}
\item[$(i)$]\ \
$H$ is one-to-one Wadge hard for $\tbD_\alpha(P)$
for every $\omega$-continuous domain $P$.

\item[$(ii)$]\ \
$H$ is one-to-one Wadge hard for $\tbD_\alpha(\PN)$.

\item[$(iii)$]\ \
$H$ is one-to-one Wadge hard for $\tbD_\alpha(P)$ for some continuous domain $P$
admitting a scattered chain of length $\alpha+1$.

\item[$(iv)$]\ \
There exists an $H$-scattered chain in $Q$ of length $\alpha+1$.
\end{enumerate}

\item
Fix some presentation $(q_n)_{n\in\N}$ of $Q$.
If $\alpha<\CK$ the following conditions are equivalent.
\begin{enumerate}
\item[$(i_e)$]\ \
$H$ is one-to-one Wadge hard for $\tbD_\alpha(P)$
and effectively one-to-one Wadge hard for $\tlD_\alpha(P)$
for every effective $\omega$-continuous domain $P$
(cf. Definition~\ref{def:effective domain}).

\item[$(ii_e)$]\ \ \ 
$H$ is effectively one-to-one Wadge hard for $\tlD_\alpha(\PN)$.

\item[$(iii_e)$]\ \ \ 
$H$ is effectively one-to-one Wadge hard for $\tlD_\alpha(P)$
for some effective $\omega$-algebraic domain $P$
admitting a computable scattered chain of length $\alpha+1$.

\item[$(iv_e)$]\ \ \
There exists an effective $H$-scattered chain in $Q$ of length $\alpha+1$.
\end{enumerate}
\end{enumerate}
\end{theorem}
\begin{proof}
$(i) \Rightarrow (ii)$.
Particularize item $(i)$ with $P=\PN$.

$(ii) \Rightarrow (iii)$.
Let $P=\PN$
and use Corollary~\ref{coro:PN A scattered}.

$(iii) \Rightarrow (iv)$.
By Proposition~\ref{p:decreasing implies A special},
the strictly decreasing chain in $P$ of length $\alpha+1$
given by $(i)$ is $A$-special for some $A\in\tbD_\alpha(P)$.
Being scattered and $A$-special, this chain is $A$-scattered.
A one-to-one continuous reduction of $A$ to $H$ maps this $A$-scattered chain
onto an $H$-scattered chain in $Q$.

$(iv) \Rightarrow (i)$.
Fix some presentation $(b_i)_{i\in\N}$ of $P$.
Let $(y_\beta)_{\beta\leq\alpha}$ be the $H$-scattered chain
and $\theta_\beta$'s, $\beta\leq\alpha$, the one-to-one continuous
maps $\PN\to P$ as in Definition~\ref{def:scattered}.
We keep the notations in the proof of the same implication
in Theorem~\ref{thm:Dalphachains}. 
Let $f:P\to Q$ be the continuous reduction of $X$ to $H$ constructed in this proof:
$f(z)=y_{\tau(z)}$ where
$\tau(z)$ is the least $\beta<\alpha$ such that $z\in B_\beta$
if $z\in\bigcup_{\beta<\alpha}V_\beta$ and $\tau(z)=\alpha$ otherwise.
We introduce another map $g:P\to Q$ such that
$g(z)=\theta_{\tau(z)}(\{i\mid b_i\ll z\})$.
Let us see that $g:P\to Q$ reduces $A$ to $H$.
If $z\in A$ then $f(z)=y_{\tau(z)}\in H$.
Suppose $\tau(z)<\alpha$.
We know that $\theta_{\tau(z)}$ maps $\PN$ into
$\{y\mid y_{\tau(z)}\sqsupseteq y\sqsupset y_{\tau(z)+1}\}$.
Since $y_{\tau(z)}\in H$, $\theta_{\tau(z)}$ maps $\PN$ into $H$.
In particular, $g(z)=\theta_{\tau(z)}(\{i\mid b_i\ll z\})\in H$.
The case $\tau(z)=\alpha$ is similar.
If  $z \notin A$ the proof can be treated similarly.

We now check that $g:P\to Q$ is one-to-one.
Suppose $z,t\in P$ and $z\neq t$.
If $\tau(z)<\tau(t)$ then
$x_{\tau(z)}\sqsupseteq g(z)            \sqsupset x_{\tau(z)+1}
                    \sqsupseteq x_{\tau(t)} \sqsupseteq g(t)$
hence $g(z) \neq g(t)$.
If $\tau(z)=\tau(t)$ then, for the same $\beta=\tau(z)$, we have
$g(z)=\theta_\beta(\{i\mid b_i\ll z\})$ and $g(t)=\theta_\beta(\{i\mid b_i\ll t\})$.
Since $z\neq t$ we have $\{i\mid b_i\ll z\}\neq\{i\mid b_i\ll t\}$.
Since $\theta_\beta$ is one-to-one we have $g(z)\neq g(t)$.
Let us now see that $g:P\to Q$ is continuous.
It is clear that $g$ is increasing.
Suppose $z=\sqcup_{n\in\N}z_n$ where $(z_n)_{n\in\N}$ is an increasing sequence
of elements of $P$.
Since $z\in V_{\tau(z)}$ and $V_{\tau(z)}$ is open, for $n$ large enough,
$z_n$ is also in $V_{\tau(z)}$  hence $\tau(z_n)\leq\tau(z)$.
Now, $\tau:P\to\{\beta\mid \beta\leq\alpha\}$ is decreasing so that
$\tau(z_n)\geq\tau(z)$.
Thus, we get equality $\tau(z_n)=\tau(z)$.
Consequently, $g(z)$ and the $g(z_n)$'s, for $n$ large enough, are of the form
$\theta_\beta(z)$ and $\theta_\beta(z_n)$ for the same $\beta=\tau(z)$.
Since $\theta_\beta$ is continuous we get
$g(z)=\theta_\beta(z)=\sqcup_{n\in\N}\theta_\beta(z_n)=\sqcup_{n\in\N}g(z_n)$.
This proves continuity of $g$.

Using Corollary~\ref{coro:PN A scattered}, the implications
$(i_e) \Rightarrow (ii_e) \Rightarrow (iii_e)  \Rightarrow (iv_e) \Rightarrow (i_e)$
are proved as in Theorem~\ref{thm:Dalphachains}.
\end{proof}
The next Corollary was known for the domain $\PN$,
(Selivanov 2005b), cf. Corollary~\ref{cor:Wadge complete exist}.

\begin{corollary}\label{cor:complete 11}
Let $Q$ be an $\omega$-continuous domain and $\alpha<\omega_1$ be a countable ordinal.
\begin{enumerate}
\item
If $Q$ admits some scattered chain of length $\alpha+1$
then there exists a one-to-one Wadge complete set for $\tbD_\alpha(Q)$.
In particular, this applies if $Q=\PN$.
\item
Suppose $Q$ is $\omega$-continuous and admits an effective presentation.
Fix such an effective presentation.
If $Q$ admits some computable scattered chain of length $\alpha+1$
then there exists a one-to-one Wadge complete set for $\tbD_\alpha(Q)$
which is also effectively Wadge complete for $\tlD_\alpha(Q)$.
In particular, this applies if $Q=\PN$.
\end{enumerate}
\end{corollary}

\begin{proof}
Using Corollary~\ref{coro:PN A scattered}
the proof is similar to that of Corollary~\ref{cor:complete D_alpha chains}.
\end{proof}
Wadge completeness does not coincide with one-to-one Wadge completeness
in $\PN$.

\begin{corollary}\label{cor:one to one}
 There are Wadge complete families for $\tbD_\alpha(\PN)$
which are not Wadge one-to-one complete.
\end{corollary}

\begin{proof}
Consider a partition $(A_\beta)_{\beta<\alpha}$ of $\N$
in infinite sets and let $A_\alpha=\emptyset$.
Let $X_\beta=\bigcup_{\beta\leq\delta\leq\alpha}A_\delta$
and $U_\beta$ be the  sets
that meet $\bigcup_{\gamma\leq\beta}A_\gamma$.
Let $\+H=D_\alpha((U_\beta)_{\beta<\alpha})$.
Since $(X_\beta)_{\beta<\alpha}$ is a decreasing $\+H$-alternating chain,
Theorem \ref{thm:Dalphachains} ensures that  $\+H$ is Wadge complete
for $\tbD_\alpha(\PN)$.
Observe that
$\bigcup_{\beta<\alpha}U_\beta=\PN\setminus\{\emptyset\}$.
Thus, every decreasing $\+H$-alternating chain of lenghth
$\alpha+1$ has the empty set as last element hence it cannot be scattered.
This shows that $\+H$ cannot be Wadge one-to-one complete.
\end{proof}

\subsection{Finite-to-one Wadge Completeness and Alternating Decreasing Chains}

In the case of the $\omega$-continuous domain $(\PN,\subseteq)$
finite-to-one and one-to-one Wadge hardness concide.

\begin{theorem}\label{thm:11 and finite to 1}
Let $\+H\subseteq\PN$.
\begin{enumerate}
\item
The following conditions are equivalent.
\begin{enumerate}
\item[$(i)$]\ \
$\+H$ is finite-to-one Wadge hard for $\tbD_\alpha(\PN)$.

\item[$(ii)$]\ \
$\+H$ is one-to-one Wadge hard for $\tbD_\alpha(P)$
for every $\omega$-continuous domain $P$.
\end{enumerate}

\item
Fix some presentation of $Q$.
If $\alpha<\CK$ the following conditions are equivalent.
\begin{enumerate}
\item[$(i_e)$]\ \
$\+H$ is effectively finite-to-one Wadge hard for $\tlD_\alpha(\PN)$.

\item[$(ii_e)$]\ \
$\+H$ is one-to-one Wadge hard for $\tbD_\alpha(P)$
and effectively one-to-one Wadge hard for $\tlD_\alpha(P)$
for every effective $\omega$-continuous domain
(cf. Definition~\ref{def:effective domain}).
\end{enumerate}
\end{enumerate}

\end{theorem}

\begin{proof}
$(i)\Rightarrow(ii)$.
Consider the family $\+A\in\tbD_\alpha(\PN)$ and the $\+A$-special chain
$(X_\beta)_{\beta\leq\alpha}$ in $\PN$
given by Proposition~\ref{p:decreasing chains PN}.
Consider bijections $\mu_\alpha:\N\to X_\alpha$ and, for $\beta<\alpha$,
$\mu_\beta:\N\to X_\beta\setminus X_{\beta+1}$.
Let $X_{\alpha,n}=\mu_\alpha(\{0,\ldots,n\})$ and,
for $\beta<\alpha$, let $X_{\beta,n}=X_{\beta+1}\cup\mu_\beta(\{0,\ldots,n\})$.
Let $f:\PN\to\PN$ be a finite-to-one continuous reduction of $\+A$ to $\+H$
and let $Z_\beta=f(X_\beta)$, for $\beta\leq\alpha$.
Consider some $\beta<\alpha$.
The sequence $(X_{\beta,n})_{n\in\N}$
is an infinite chain in $\PN$ of sets containing $X_{\beta+1}$ and contained in $X_\beta$.
Since $f$ is increasing we have
$Z_\beta=f(X_\beta)\supseteq f(X_{\beta,n}) \supseteq f(X_{\beta+1})=Z_{\beta+1}$.
Since $f$ is finite-to-one,
$(f(X_{\beta,n}))_{n\in\N}$ contains infinitely many distinct sets.
In particular, $Z_\beta\setminus Z_{\beta+1}$ must be an infinite set.
A similar argument shows that $Z_\alpha$ is also an infinite set.
Let $\theta_\alpha:\PN\to\PN$ be such that
$\theta_\alpha(X)=\mu_\alpha(X)$ for all $X\in\PN$.
For $\beta<\alpha$, let $\theta_\beta:\PN\to\PN$ be such that
$\theta_\alpha(X)=f(X_{\beta+1}\cup\mu_\beta(X))$.
These maps $\theta_\beta$, $\beta\leq\alpha$,
are continuous one-to-one maps $\PN\to\PN$.
Clearly, $\theta_\alpha(\emptyset)=\emptyset$ and $\theta_\alpha(\N)=X_\alpha$.
And for $\beta<\alpha$,
$\theta_\beta(\emptyset)=Z_{\beta+1}$ and $\theta_\beta(\N)=Z_\beta$.
Recall that $(X_\beta)_{\beta\leq\alpha}$ is $\+A$-special.
Thus, for all $X\in\PN$,
the set $\mu_\alpha(X)$, being included in $X_\alpha$, is not in $\+A$.
For $\beta<\alpha$, 
the set $X_{\beta+1}\cup\mu_\beta(X)$, being included in $X_\beta$
and containing strictly $X_{\beta+1}$, is in $\+A$ if and only if~$\beta\not\sim\alpha$.
Since $f$ is a reduction of $\+A$ to $\+H$, 
the range of $\theta_\alpha$ is included in $\PN\setminus\+H$
and, for $\beta<\alpha$, the range of $\theta_\beta$ is included in $\+H$ if $\beta\not\sim\alpha$,
and it is included in $\PN\setminus\+H$ if $\beta\sim\alpha$.
This proves that the chain $(Z_\beta)_{\beta\leq\alpha}$ is $\+H$-scattered.
Applying Theorem~\ref{thm:one to one}, we get the wanted condition ($ii$).

$(i_e)\Rightarrow(ii_e)$. Routine variation of the above argument.

$(ii)\Rightarrow(i)$ and $(ii_e)\Rightarrow(i_e)$ are obvious.
\end{proof}

\subsection{Increasing Chains of Wadge Degrees}

Combining a construction given by Selivanov (2005b)  
with decreasing chains of sets, we get increasing
chains of proper $\tbD_\alpha(\PN)$ Wadge degrees.
We write $\N^{<\omega}$ for the set of finite sequences of natural numbers, 
and $\Bmix$ for $\N^{<\omega}\cup\Baire$.
 $\nil$ denotes the empty sequence.

\begin{definition}\label{def:BT}
\cite{selivanov2005}
Suppose $\alpha$ is infinite.
Let $T\subset\Bmix$ be a well-founded tree
(i.e. $T$ is closed under prefix and has no infinite branch)
with rank $\alpha$.
Let $\xi:\N^{<\omega}\to2\N$ be an injective map such that,
$\xi(\nil)=0$.
Let $e:\N^{<\omega}\to\+P_{<\omega}(2\N)$ be the map such that,
for every $\sigma\in T$,
$e(\sigma)=\{\xi(\tau)\mid \tau\textit{ prefix of }\sigma\}$.
Let
$B(T)=\{X\in\PN\mid \forall\sigma\in T\ \ X\not\subseteq e(\sigma)\}.$
And, for each ordinal $\alpha$, let 
$Y_\alpha=\{e(\sigma)\mid\sigma\in T\textit{ has odd length }\}$ and 
$Z_\alpha=Y_\alpha\cup B(T)$.
\end{definition}

\begin{theorem}\label{thm:Yalpha}
(Selivanov 2005b, Lemma 5.8 and Propositions 5.9 and 6.4)
$B(T)$ is open and 
$Y_\alpha,Z_\alpha$ are two proper $\tbD_\alpha(\PN)$ families
that are Wadge incomparable.
\end{theorem}
Before entering the construction of the chains of Wadge degrees,
a simple lemma.

\begin{lemma}\label{l:Dalpha operations}
Let $E$ be a topological space.
\begin{enumerate}
\item  If $A\in\tbD_\alpha(E)$
and $L\subseteq E$ is open
then $A\cap L\in\tbD_\alpha(E)$.

\item
 If $\alpha$ is infinite, $A\in\tbD_\alpha(E)$
and $F\subseteq E$ is closed
then $A\cap F\in\tbD_\alpha(E)$.

\item
 Suppose $\beta>0$, $\beta+\alpha=\alpha$
(i.e. $\alpha\geq\beta\omega$).
If $C\in\tbD_\beta(E)$,
$D\in\tbD_\alpha(E)$,
$L$ is open and $C\subseteq L$
and $L$ and $D$ are disjoint
then $C\cup D\in\tbD_\alpha(E)$.
\end{enumerate}
\end{lemma}

\begin{proof}
1. If $A=D_\alpha((U_\beta)_{\beta<\alpha})$
then $A\cap L=D_\alpha((U_\beta\cap L)_{\beta<\alpha})$.

2. Let $L=E\setminus F$.
Let $A=D_\alpha((U_\beta)_{\beta<\alpha})$
where the $U_\beta$'s are open.
In case $\alpha$ is odd,
let $X_0=\emptyset$, $X_1=L$ and,
for $\beta<\alpha$, let $X_{2+\beta}=L\cup U_\beta$.
In case $\alpha$ is even, let $X_0=L$ and
$X_{1+\beta}=L\cup U_\beta$.
Now, $2+\alpha=1+\alpha=\alpha$ (since $\alpha$ is infinite)
and $A\cap F=D_\alpha((X_\beta)_{\beta<\alpha})$.

3. Suppose $C=D_\beta((U_\gamma)_{\gamma<\beta})$
and $D=D_\alpha((V_\delta)_{\delta<\alpha})$.
{\it Case $\beta\sim\alpha$.}
Let $W_\gamma=U_\gamma\cap L$ for $\gamma<\beta$,
$W_{\beta+\delta}=V_\delta\cup L$ for $\delta<\alpha$.
Then,

$\begin{array}{rcl}
D_\alpha((W_\gamma)_{\gamma<\alpha})
&=&D_\beta((U_\gamma\cap L)_{\gamma<\beta})
\ \ \cup\ \
(D_\alpha((V_\delta\cup L)_{\beta\leq\delta<\alpha})
         \setminus\bigcup_{\gamma<\beta}U_\gamma\cap L)
\\&=&(C\cap L)
\ \ \cup\ \ ((D\setminus L)
\setminus\bigcup_{\gamma<\beta}U_\gamma\cap L)
\\&=& (C \cap L) \cup (D \setminus L)
\\&=&C\cup D.
\end{array}
$
\\
{\it Case $\beta\not\sim\alpha$.}
Let $W_0=\emptyset$,
$W_{1+\gamma}=U_\gamma\cap L$ for $\gamma<\beta$,
$W_{\beta+\delta}=V_\delta\cup L$ for $\delta<\alpha$,
and argue~similarly.
\end{proof}

\begin{definition}\label{def:Yalphabeta}
Let $(a_\delta)_{\delta<\alpha}$ be a one-to-one enumeration
of $2\N+3$.
For $\beta<\alpha$ let $U_\beta$ be the open family of sets
containing $1$ and meeting
$\{a_\gamma\mid\gamma\leq\beta\}$.
Let $\+A=D_\alpha((U_\delta)_{\delta<\alpha})
\in\tbD_\alpha(\PN)$.
For $\beta<\alpha$, set
$\+A_\beta=D_\beta((U_\delta)_{\delta<\beta})\cap\+P(2\N+1)$.
Let $Y_\alpha$  as in Definition \ref{def:BT} and 
let $Y_{\alpha,\beta}=\+A_\beta\cup Y_\alpha$
and 
$Z_{\alpha,\beta}=(B(T)\cap\+B_{\{0\}})\cup Y_{\alpha,\beta}$.
\end{definition}
The following result is straightforward.

\begin{proposition}\label{p:Yalphabeta}
\begin{enumerate}

\item
 $Y_{\alpha,\beta}$ is included in $\+P(2\N)\cup\+P(2\N+1)$. 

\item
 A set $X$ is in $\+A$ if and only if $1\in X$
and $X$ meets $2\N+3$
and the least $\delta$ such that
$a_\delta\in X$ is such that $\delta\not\sim\alpha$.

\item
 A set $X$ is in $\+A_\beta$ if and only if $1\in X$,
$X\subseteq2\N+1$
and $X$ meets $\{a_\delta\mid\delta<\beta\}$
and the least $\delta$ such that
$a_\delta\in X$ is such that $\delta\not\sim\beta$.

\item
 For $\beta\leq\alpha$, set
$X_\beta=\{a_\delta\mid\delta\geq\beta\}\cup\{1\}$.
The sequence $(X_\delta)_{\delta\leq\alpha}$ is a decreasing
$\+A$-alternating chain.
The sequence $(X_\delta)_{\delta\leq\beta}$ is alternating
relative to $Y_{\alpha,\beta}$ and to $Z_{\alpha,\beta}$.

\item
 If $X\in B(T)$ (cf. Definition \ref{def:BT})
then $(X\cap2\N)\cup I\in B(T)$
for every $I\subseteq2\N+1$.
\end{enumerate}
\end{proposition}

\begin{theorem}\label{thm:chain degrees}
\begin{enumerate}

\item  If $\omega\leq\beta+\alpha=\alpha$
then $Y_{\alpha,\beta}$ and $Z_{\alpha,\beta}$
are proper families in $\tbD_\alpha(\PN)$.

\item 
 If $\omega\leq\beta<\gamma<\alpha$
then $Y_{\alpha,\beta}<_W Y_{\alpha,\gamma}$
and $Z_{\alpha,\beta}<_W Z_{\alpha,\gamma}$.

\item  If $\omega\leq\beta,\delta$ and $\beta,\delta<\alpha$ then
$Y_{\alpha,\beta}$ and $Z_{\alpha,\delta}$
are Wadge incomparable.

\item  For $\alpha<\CK$, 
$Y_{\alpha,\beta}$ and $Z_{\alpha,\delta}$
are proper families  in $\tlD_\alpha(\PN)$, 
and the effective counterparts of items 2 and 3  hold.
\end{enumerate}
\end{theorem}

\begin{proof}
1. First, we show that $Y_{\alpha,\beta}\in\tbD_\alpha(\PN)$.
Since $\+P(2\N+1)$ is closed in $\PN$,
item 2 of Lemma~\ref{l:Dalpha operations} ensures that
$\+A_\beta\in\tbD_\beta(\PN)$.
Now, $\+A_\beta$ is included in the open family $\+B_{\{1\}}$
which is disjoint from $Y_\alpha$ since all sets in $Y_\alpha\subset\+P(2\N)$.
Item 3 of Lemma~\ref{l:Dalpha operations} ensures that
$Y_{\alpha,\beta}=\+A_\beta\cup Y_\alpha\in\tbD_\alpha(\PN)$.
We now prove that $Z_{\alpha,\beta}\in\tbD_\alpha(\PN)$.
The family $Y_\alpha$ is disjoint from $B(T)$ hence from $B(T)\cap\+B_{\{0\}}$.
Also, $B(T)\cap\+B_{\{0\}}$ is disjoint from $\+A_\beta$
since every set in $B(T)\cap\+B_{\{0\}}$
contains $0$ whereas all sets in $\+A_\beta$ are included in $2\N+1$.
Thus, $B(T)\cap\+B_{\{0\}}$ is disjoint from $Y_{\alpha,\beta}$.
Since $B(T)\cap\+B_{\{0\}}$ is open,
applying again item 3 of Lemma~\ref{l:Dalpha operations}, we see that
$Z_{\alpha,\beta}=(B(T)\cap\+B_{\{0\}})
\cup Y_{\alpha,\beta}$ is in $\tbD_\alpha(\PN)$.
Finally, notice that the map 
$f \colon X \mapsto (X \cap 2 \mathbb{N}) \cup \{ 2n+3 \mid 2n+1 \in X \}$
is continuous and reduces $Y_\alpha$ to $Y_{\alpha,\beta}$.
The map sending 
$X$ to $f(X)$ if $X \notin B(T)$ and 
to $f(X) \cup \{ 0 \}$ if $X \in B(T)$
is also continuous and reduces $Z_\alpha$  to $Z_{\alpha,\beta}$.
Since $Y_\alpha$ and $Z_\alpha$ are outside $\co\tbD_\alpha(\PN)$ 
(cf. Theorem~\ref{thm:Yalpha}) so are $Y_{\alpha,\beta}$ and $Z_{\alpha,\beta}$.

2.  The reducibilities $Y_{\alpha,\beta} \leq_W Y_{\alpha,\gamma}$ and
$Z_{\alpha,\beta} \leq_W Z_{\alpha,\gamma}$
can be both witnessed by a same continuous map $f:\PN\to\PN$,
$$
\quad \left\{\begin{array}{lrcl}
\textit{if }\beta \not\sim \gamma\textit{ then }
&f(X)&=& 
(X \cap 2\N) \cup \{ a_{\delta+1} \mid \delta < \beta \textit{ and }a_\delta\in X \}
                   \cup(X\cap \{ 1 \})
\\
\textit{if }\beta \sim \gamma\textit{ then }
&f(X)&=& 
(X \cap 2\N) \cup (X \cap (\{ a_\delta \mid \delta < \beta \} \cup \{ 1 \})) \cup M,
\\&\multicolumn{3}{l}{\textit{\hspace{1cm }
where $M = \{ a_\beta \}$ if $X$ meets $2\N+1$
and $M = \emptyset$ otherwise.}}
\end{array}\right.
\qquad
$$
We now show that the reducibilities are strict.
Towards a contradiction, suppose that there is a reduction $f$
from $Y_{\alpha,\gamma}$ to $Y_{\alpha,\beta}$.
Then the decreasing $Y_{\alpha,\gamma}$-alternating chain
$(X_\delta)_{\delta\leq\gamma}$
would be mapped onto a $Y_{\alpha,\beta}$-alternating chain
$(f(X_\delta))_{\delta\leq\gamma}$
(recall that by item 4 of Proposition~\ref{p:Yalphabeta},
$X_\beta=\{a_\delta\mid\delta\geq\beta\}\cup\{1\}$ ).
Let $\varepsilon\in\{0,1\}$ be least such that
$f(X_\varepsilon)\in Y_{\alpha,\beta}$.
Since $Y_{\alpha,\beta}\subset\+P(2\N)\cup\+P(2\N+1)$,
we have $f(X_\varepsilon)\in\+P(2\N)\cup\+P(2\N+1)$.
The chain being strictly decreasing, we see that
the family $\{f(X_\delta)\mid \varepsilon\leq\delta\leq\gamma\}$
is included either in $\+P(2\N)$ or in $\+P(2\N+1)$. 
Now, $Y_{\alpha,\beta}\cap\+P(2\N)=Y_\alpha$
contains no infinite decreasing chain.
Thus, $\{f(X_\delta)\mid \varepsilon\leq\delta\leq\gamma\}\subseteq\+P(2\N+1)$.
Since the chain $(f(X_\delta))_{2\leq\delta\leq\gamma}$
is also $Y_{\alpha,\beta}$-alternating and has length $\gamma+1$
(since $\gamma$ is infinite),
this contradicts the fact that $Y_{\alpha,\beta}\cap\+P(2\N+1)=\+A_\beta$
does not contain alternating chains longer than $\beta+1$.
To see that $Z_{\alpha,\beta}<_W Z_{\alpha,\gamma}$,
observe that $B(T)$ is an upset, hence it has no infinite decreasing
alternating chain.
Then argue as for $Y_{\alpha,\gamma}\not<_W Y_{\alpha,\beta}$.

3. Let us check that $Y_{\alpha,\beta}\not\leq_W Z_{\alpha,\delta}$.
Towards a contradiction, suppose that $f$ reduces $Y_{\alpha,\beta}$
to $Z_{\alpha,\delta}$.
Since $\N\notin Y_{\alpha,\beta}$ we have
$f(\N)\notin Z_{\alpha,\beta}$.
In particular, $f(\N)\notin B(T)\cap\+B_{\{0\}}$.
Hence, either $f(\N)\notin B(T)$ or $0\notin f(\N)$.
Suppose $f(\N)\notin B(T)$. Then $f(\N)\subseteq e(\sigma)$
for some $\sigma\in T$. In particular, $f(\N)$ is finite hence has finitely
many subsets.
This contradicts the fact that the infinite decreasing
$Y_{\alpha,\beta}$-alternating chains are mapped by f to infinite decreasing
$Z_{\alpha,\delta}$-alternating chains below $f(\N)$.
If, instead, $0 \notin f(N)$, then $0 \notin f(X)$ (for all $X$)
hence $f$ reduces $Y_{\alpha,\beta}$ to $Y_{\alpha,\delta}=\+A_\delta\cup Y_\alpha$
hence to $\+A_\delta$
(since all sets in $Y_\alpha$ contain $0$ whereas no $f(X)$ does).
This is already a contradiction since $\delta<\alpha$,
$Y_{\alpha,\beta}$ is a proper $\tlD_\alpha(\PN)$ set by item 1, 
and $A_\delta$ is a
$\tlD_\delta(\PN)$ set by the proof of item 1.
Finally, we check that $Z_{\alpha,\delta}\not\leq_W Y_{\alpha,\beta}$.
Since $\N\in Z_{\alpha,\delta}$ (it is in $B(T)$) we have
$f(\N)\in Y_{\alpha,\beta}$.
If $f(\N)\in Y_\alpha$ then $f(\N)$ is finite.
This contradicts the existence of infinite alternating decreasing chains.
Thus, $f(\N)\in \+A_\beta$ hence $f(\N)$ is included in $2\N+1$.
Since $f$ is increasing, $f$ takes values in $\+P(2\N+1)$.
Thus, $f$ reduces $Z_{\alpha,\delta}$ to $A_\beta$.
But $\beta<\alpha$, $Z_{\alpha,\delta}$ is a proper $\tlD_\alpha(\PN)$ set by item 1, 
and $A_\beta$ is a $\tlD_\beta(\PN)$ set by the proof of item 1,
 the same contradiction as before.

4. Since $\alpha<\CK$, the tree $T$ in Definition~\ref{def:BT} can be taken computable.
Then $Y_\alpha$ is a computable subfamily of $\+P_{<\omega}\N$
and $B(T)$ is an effective open family.
We can take a computable enumeration $(a_\delta)_{\delta<\alpha}$
as in Definition~\ref{def:Yalphabeta}.
Thus, the $Y_{\alpha,\beta}$'s and $Z_{\alpha,\gamma}$'s are in $\tlD_\alpha(\PN)$.
To see that for $\beta<\gamma$,
$Y_{\alpha,\beta}$ is effectively reducible to $Y_{\alpha,\gamma}$, and
$Z_{\alpha,\beta}$ is effectively reducible to $Z_{\alpha,\gamma}$,
observe that the map $f$ defined
in the proof of item 2 is effective.
\end{proof}

\section{Wadge Hardness in $\PN$ Versus Wadge Hardness in Other Spaces}
\label{s:PN versus}

Due to the universality property of $\PN$ (cf. Proposition~\ref{p:Baire PN universal}),
Wadge hardness relative to a Borel class of $\PN$  is stronger
than Wadge hardness relative to the same Borel class of many other spaces.

\subsection{Simple Wadge Complete Sets in $\PN$ for $\tbSigma^0_2$}

\begin{definition}\label{def:S2} 
Let\
$\+S_2=\{X\in\PN \mid \exists n\ (2n\in X\wedge 2n+1\notin X)\}$.
\end{definition}

\begin{proposition}\label{p:Si} 
The family $\+S_2$ is Wadge complete for $\tbSigma^0_2(\PN)$
and effectively Wadge complete for $\Sigma^0_2(\PN)$.
\end{proposition}

\begin{proof}
From its definition it is clear that $\+S_2$ is $\Sigma^0_2(\PN)$.
Consider a $\tbSigma^0_2(\PN)$ family
$\+X=\bigcup_{n\in\N} \{Z\mid\exists A\in\+A_n\ Z\supseteq A\}
                \setminus \{Z\mid\exists B\in\+B_n\ Z\supseteq B\}$,
where the $\+A_n$'s and $\+B_n$'s are families of finite sets.
Define $f:\PN\to\PN$ as
$$
f(X)\ =\ \{2n\mid\exists A\in\+A_n\ A\subseteq X\}
\  \cup\ \{2n+1\mid\exists B\in\+B_n\ B\subseteq X\}.
$$
Then, $X\in\+X \ \Leftrightarrow \ 
\exists n\ (2n\in f(X)\wedge 2n+1\notin f(X)) 
\ \Leftrightarrow \ f(X)\in\+S_2$.
Observe that $f$ is continuous hence it is  a reduction of $\+X$ to $\+S_2$.
In case $\+X$ is $\Sigma^0_2(\PN)$ then
the $\+A_n$, $\+B_n$ are uniformly computably enumerable
so that $f$ is effective.
\end{proof}

\subsection{Alternating Increasing Chains for $\tbSigma^0_2(\PN)$ Families}

As stated in Remark~\ref{rk:no increasing alternating},
in $\omega$-continuous domains  there are no increasing
$A$-alternating chains of length $\omega$ when $A$ is $\tbDelta^0_2$.
This is not  true anymore for $A$ in $\tbSigma^0_2$.

\begin{proposition}\label{p:long chain}
For every ordinal $\alpha<\omega_1$, there exists in $\PN$ an
$\+S_2$-alternating  increasing chain with length $\alpha$.
If $\alpha<\CK$ then the sequence can be taken computable.
\end{proposition}

\begin{proof}
Let $(a_\beta)_{\beta<\alpha}$ be an injective enumeration of $\N$.
Let
\smallskip\\
\centerline{$X_{2\beta}=\{2a_\gamma,2a_\gamma+1\mid\gamma<\beta\}\cup\{2a_\beta\}
\text{ and } X_{2\beta+1}=X_{2\beta}\cup\{2a_\beta+1\}$.}
\smallskip
Clearly, $(X_\beta)_{\beta<\alpha}$ is increasing,
and for each $\beta<\alpha$,  
$X_{2\beta}\in\+S_2$ but  $X_{2\beta+1}$ is not.
\end{proof}

\begin{remark}
As a corollary, we see that the  space $\mix$,
where all increasing chains have length at most $\omega+1$, 
contains no Wadge hard set for $\tbSigma^0_2(\PN)$.
\end{remark}
The above proposition can be  improved.
We first define a  notion of alternating chain.

\begin{definition}\label{def:alt chain}
Let $(E,\leq)$ be a linear ordering and $A\subseteq E$.
Let $(P,\sqsubseteq)$ be a partially ordered set and $H\subseteq P$.
\begin{enumerate}
\item
$A$ is an alternation on $E$ if, for all $a,b\in E$, the following conditions are satisfied.

\begin{itemize}
\item  If $b$ is successor to $a$
then $a\in A$ if and only if~$b\notin A$.

\item  If $a$ has no predecessor nor successor then $a\notin A$.
\end{itemize}
\item 
An $(H,A)$-alternating $E$-chain in $P$ is an $E$-sequence $(x_a)_{a\in E}$ of elements of $P$
that satisfies the following conditions: for all $a,b\in E$,
\begin{itemize}
\item  $a\leq b$ (in the linear order $E$) if and only if $x_a\sqsubseteq x_b$ (in the partial order $P$),

\item  $x_a\in H$ if and only if $a\in A$.
\end{itemize}
\end{enumerate}
\end{definition}

We now define a very rich  linear order with respect to alternation.

\begin{definition} \label{def:worst}
Let $(\+R,\leq)$ be the totally ordered set obtained from $\R$
by replacing each rational number by a chain of type $\Z$.
Formally, $\+R=(\R\setminus \Q)\cup(\Q\times \Z)$
with the following order:
for $x \in \R\setminus \Q$, $q,q'\in\Q$, $z,z'\in\Z$,
$$
x<(q,z) \iff x<q, \ \
(q,z)< x\iff q<x, \ \
(q,z)<(q',z') \iff q<q'  \mbox{ or }  (q=q' \mbox{ and } z<z')\ .
$$
\end{definition}

\begin{remark}\label{rk:worst}
The above linear order is 
very rich in the following sense.
In a linear order $(E,\leq)$ let $\sim_E$ 
be the equivalence 
relation such that $x\sim_E y$
if and only if there are finitely many elements between $x$ and $y$
(i.e. we identify an element with its successor if there is one).
If $(E,\leq)$ has countably many successor pairs and its quotient is embedded in $\R$
then $(E,\leq)$ is embedded in $\+R$. 
\end{remark}

\begin{proposition}\label{p:worst}
Le $P$ be an $\omega$-continuous domain and $H\subseteq P$.
\begin{enumerate}
\item
If $H$ is Wadge hard for $\tbSigma^0_2(\PN)$
then $P$ contains an $(H,\+A)$-alternating $\+R$-chain
for every alternation $\+A$ on $\+R$ (cf. Definition~\ref{def:worst}).
In particular, the set $P\setminus H$ has cardinality $2^{\aleph_0}$. 
\item
Suppose $P$ is an $\omega$-continuous domain
and fix some presentation $(p_i)_{i\in\N}$ of $P$.
If $H\subseteq P$ is effectively Wadge hard for $\Sigma^0_2(\PN)$
then,
for every alternation $\+A$ on $\+R$ which is computable
as a subset of $\Q\times\Z$,
$P$ contains an  $(H,\+A)$-alternating $\+R$-chain $(u_\xi)_{\xi\in\+R}$ such that
the sets
$$
\{(q,n,i)\mid q\in\Q, n\in\Z, p_i\ll u_{(q,n)}\}
\text{ \ and \ }
\{(q,i)\mid q\in\Q, \forall x\in\R\setminus\Q\ (q<x\Rightarrow p_i\ll u_x)\}
$$ 
are both computably enumerable.
In particular, the set $P\setminus H$ has cardinality $2^{\aleph_0}$. 
\end{enumerate}
\end{proposition}

\begin{note}\label{note:worst}
It is known that any Borel subset of a Polish space
is either countable or has cardinality $2^{\aleph_0}$.
The same is true in any quasi-Polish space,
hence in any $\omega$-continuous domain, since its Borel sets
are those of a Polish topology on the same space
(de Brecht 2011).
Since the Cantor space $\cantor$ is $T_1$, any singleton set is closed.
Hence, any set with countable complement is in $\tbPi^0_2(\cantor)$
and it cannot be Wadge hard for $\tbSigma^0_2(\cantor)$.
This easy argument, showing that the complement of any set 
$ A \subseteq \cantor$  which is Wadge hard for 
$\tbSigma^0_2(\cantor)$  has cardinality $2^{\aleph_0}$,
breaks down in the Scott domain $\PN$. If $X\in \PN$, $X\not=\emptyset$ then the
singleton family $\{X\}$ is in $\tbPi^0_2(\PN)$
but it is not closed in $\PN$. Thus,  a family with a countable complement need not
to be in  $\tbPi^0_2(\PN)$.
\end{note}

\begin{proof}[Proof of Proposition~\ref{p:worst}.]
1. For every alternation $\+A$ on $\+R$,
we construct in $\PN$ an $(\+S_2,\+A)$-alternating $\+R$-chain
$(X_\xi)_{\xi\in\+R}$.
For every $q\in\Q$, let $\varepsilon_q=0$ if $(q,0)\in \+A$
and  $\varepsilon_q=1$ if $(q,0)\notin \+A$.
Let $(E_q)_{q\in\Q}$ be a partition of $2\N$ into infinite sets.
For every $q\in\Q$, let $\varphi_q:\Z\to E_q$ be some bijective map.
The $\+R$-chain is defined as follows.
\\
For $q\in\Q, n\in\Z$, $x\in\R\setminus \Q$,
\\
\centerline{$\begin{array}{rcl}
X_{(q,2n+\varepsilon_q)} &=& (\bigcup_{r\in\Q, r<q}  E_r)
\cup
\{\varphi_q(p), \varphi_q(p)+1 \ \mid\ p\in\Z, p<n\}
\\
X_{(q,2n+\varepsilon_q+1)} &=& X_{(q,2n)} \cup  \{\varphi_q(n)\} 
\\
X_x  &=& \bigcup_{q\in\Q, q<x}  E_q
\end{array}$}
\\
It is easy to check that this chain is $(\+S_2,\+A)$-alternating.
Suppose now that $H\subseteq P$ is Wadge hard for
$\tbSigma^0_2(\PN)$. 
Let $f:\PN\to P$ be continuous such that
$f^{-1}(H)=\+S_2$.
Then $f$ maps this $(\+S_2,\+A)$-alternating $\+R$-chain
onto an $(H,\+A)$-alternating $\+R$-chain in~$P$.
If $x,y\in\R\setminus\Q$ and $x<y$ then
$X_x\subset X_{q,0}\subset X_{q,1}\subset X_y$.
Applying the  reduction $f$, we get
$f(X_x)\sqsubseteq f(X_{q,0})\sqsubseteq f(X_{q,1})\sqsubseteq f(X_y)$.
Since $X_{q,0}\in\+S_2\iff X_{q,1}\notin\+S_2$ we have
$f(X_{q,0})\in H\iff f(X_{q,1})\notin H$
hence $f(X_{q,0})\sqsubset f(X_{q,1})$ and $f(X_x)\sqsubset f(X_y)$.
Thus, $P\setminus H$ has cardinality $2^{\omega}$ since it contains the 
pairwise distinct elements $f(X_x)$'s,~$x\in\R\setminus\Q$.
  
2. Observe that the map $q\mapsto\varepsilon_q$ is computable
and take a computable partition $(E_q)_{q\in\Q}$ and a computable sequence
of bijections~$\varphi_q$.
\end{proof}

\subsection{Classical Families Hard for $\tbSigma^0_2(\cantor)$
but Not Hard for $\tbSigma^0_2(\PN)$}

A priori, modulo the obvious
bijection between the Cantor space $\cantor$ and $\PN$,
families of $\cantor$
that are Wadge hard for a class in $\cantor$ may seem good candidates
to be also hard families for the same class  in~$\PN$.
This is not the case.
We state here some counterexamples.
For results about the Wadge hardness of classical $\tbSigma^0_3(\PN)$  families 
for various classes and spaces, we refer the reader to \cite{jsl2}.

\begin{theorem}\label{thm:FIN}
Let $\theta:\PN\to\cantor$ be the bijection  which maps a set to its characteristic function.
Then,
\begin{enumerate}
\item $\theta(\FIN)$ is Wadge $\tbSigma^0_2(\cantor)$-complete
and effectively Wadge $\Sigma^0_2(\cantor)$-complete.
In particular, $\theta(\FIN)$ is in
$\Sigma^0_2(\cantor)\setminus\tbPi^0_2(\cantor)$.

\item 
$\FIN\in\Sigma^0_2(\PN)\setminus\tbPi^0_2(\PN)$.

\item
$\FIN$ is Wadge hard for $\tbSigma^0_2(\cantor)$ 
and for $\tbPi^0_1(\PN)$.
\item
$\FIN$ Wadge reduces no non-empty
$\tbSigma^0_1(\PN)$ family.
\item
No subfamily of $\FIN$ is Wadge hard for $\tbSigma^0_2(\PN)$.
\end{enumerate}
Items  3,4,5 hold with effective versions.
\end{theorem}

\begin{proof}
1. Classical result \cite{kechrisBook},
$\theta(\PN)$ is just the
set of all sequences eventually equal to~$0$.
 
2. We use an argument from \cite{selivanov2005}.
If $\FIN$ were $\tbPi^0_2(\PN)$ then
the family of infinite sets would be $\tbSigma^0_2(\PN)$,
hence a countable union of differences
of open sets $\bigcup_{n\in\N}\+U_n\setminus\+V_n$.
If $X\in\+U_n\setminus\+V_n$ then there exists $A$ finite
such that $A\subset X$ and $A\in\+U_n$.
Since $X\notin\+V_n$ which is an upset,  also $A\notin\+V_n$.
Thus, $\+U_n\setminus\+V_n$ contains a finite set, contradiction.

3. That $\FIN$ is Wadge hard for $\tbSigma^0_2(\cantor)$
follows easily from the fact that the set of sequences
eventually equal to $0$ is $\tbSigma^0_2(\cantor)$-complete,
and the continuity of 
the map $\theta^{-1}:\cantor\to\PN$  (indeed it  is an effective map).
For the second assertion, any family $\+Y$ in $\tbPi^0_1(\PN)$  can be written
$\+Y=\{Z \mid \forall A\in\+A\ A\not\subseteq Z\}$
where $\+A$ is a subfamily of $\FIN$.
Define $f:\PN\to\PN$ such that
$f(Z)=\N$ if $Z$ contains some $A\in\+A$
and $f(Z)=\emptyset$, otherwise.
Then, $f$ is a reduction of $\+X$ to~$\FIN$.

4. Since $\FIN$ is closed under subsets, and continuous reductions
$\PN\to\PN$ are  increasing,
every family $\+X\subset\PN$ which is Wadge reducible to $\FIN$
is closed under subsets.
In particular, no non-empty $\Sigma^0_1(\PN)$ family admits a
continuous reduction 
to~$\FIN$.

5. Observe that  $\FIN$ contains only chains of  length $\omega$ and
apply Proposition~\ref{p:worst}.
\end{proof}
\bigskip
{\bf Acknowledgements}:  The authors are indebted to an anonymous referee
for a huge number of improvements in our original manuscript.


\end{document}